\documentclass[5p,sort&compress]{elsarticle}   

\usepackage{amsmath} 
\usepackage{amssymb}  

\usepackage{graphicx}          

\usepackage{tabulary}
\usepackage[english]{babel}
\usepackage{blindtext}
\usepackage[makeroom]{cancel}
\usepackage{amsfonts}
\usepackage{amsthm}  

\graphicspath{{../pdf/}{../jpeg/}}
\DeclareGraphicsExtensions{.pdf,.jpeg,.png}
\usepackage{pgf,tikz}
\usepackage{tabularx}
\usepackage{float}
\usepackage{subcaption}

\usepackage{enumitem} 
\usepackage{bbm} 
\usepackage{url}
\usepackage{algpseudocode}
\usepackage{algorithm}

\usepackage[strict]{changepage}

\newcommand{\vect}[1]{\boldsymbol{#1}}
\newcommand{\mat}[1]{\boldsymbol{#1}}

\DeclareMathOperator{\diag}{\text{diag}}

\renewcommand{\eqref}[1]{Eq.~(\ref{#1})}  

\newdefinition{remark}{Remark}
\newtheorem{theorem}{Theorem}
\newtheorem{lemma}{Lemma}
\newtheorem{assumption}{Assumption}

\newtheorem{proposition}{Proposition}

\newdefinition{problem}{Problem}

\DeclareMathAlphabet{\dsmath}{U}{BOONDOX-ds}{m}{n}

\begin{document}

\begin{frontmatter}

\title{Decentralised adaptive-gain control for eliminating epidemic spreading on networks\tnoteref{t2}} 

\tnotetext[t2]{M. Ye is supported by the Western Australian Government through the Premier's Science Fellowship Program and the Defence Science Centre. Z. Sun is supported by a starting grant from Eindhoven Artificial Intelligence Systems Institute (EAISI), the Netherlands.}

\author[CUR]{Liam Walsh}
\author[CUR]{Mengbin Ye\corref{cor1}}
\author[ANU]{Brian D.O. Anderson}
\author[TUE]{Zhiyong Sun}

\cortext[cor1]{Corresponding author. Email: \texttt{mengbin.ye@curtin.edu.au}}

\address[CUR]{Centre for Optimisation and Decision Science, Curtin University, Perth, Australia}  
\address[ANU]{School of Engineering, Australian National University, Canberra, Australia}             
\address[TUE]{Department of Electrical Engineering, Eindhoven University of Technology, The Netherlands}        


\begin{abstract}      
This paper considers the classical Susceptible--Infected--Susceptible (SIS) network epidemic model, which describes a disease spreading through $n$ nodes, with the network links governing the possible transmission pathways of the disease between nodes. 
We consider feedback control to eliminate the disease in scenarios where the disease would otherwise persist in an uncontrolled network. 
We propose a family of decentralised adaptive-gain control algorithms, in which each node has a control gain that adaptively evolves according to a differential equation, independent of the gains of other nodes. The adaptive gain is applied multiplicatively to either decrease the infection rate or increase the recovery rate. To begin, we assume all nodes are controlled, and prove that both infection rate control and recovery rate control algorithms eliminate the disease with the limiting gains being positive and finite. Then, we consider the possibility of controlling a subset of the nodes, for both the infection rate control and recovery rate control. We first identify a necessary and sufficient condition for the existence of a subset of nodes, which if controlled would result in the elimination of the disease. 
For a given network, there may exist several such viable subsets, and we propose an iterative algorithm to identify such a subset. Simulations are provided to demonstrate the effectiveness of the various proposed controllers.  
\end{abstract}

\end{frontmatter}

\section{Introduction}


Mathematical models of epidemics have been used for over a century to study the spread of infectious diseases in a population~\cite{zino2021analysis}. The deterministic Susceptible--Infected--Susceptible (SIS) model is a classical paradigm~\cite{brauer2008mathematical}; the single population model posits that each individual in the population exists in the mutually exclusive health compartments of i) healthy and susceptible to the disease (S), and ii) infected and able to transmit the disease (I). An infected individual can transmit the disease to a susceptible individual, and infected individuals can recover from the disease; it is assumed that recovery provides no immunity to re-infection. The Susceptible--Infected--Removed (SIR) model adds a third ``removed'' health compartment, for diseases which grant recovered individuals with permanent immunity from reinfection or to capture dead individuals~\cite{brauer2008mathematical}. 
Networked SIS and SIR models have since been proposed, whereby each node in the network represents a population, while edges between nodes represent pathways for the disease to spread between populations~\cite{brauer2008mathematical,mei2017epidemics_review}.

Governments have increasingly used epidemic models to inform public health measures and strategies to control the spread of epidemics over the last century, including notably the COVID-19 pandemic~\cite{zino2021analysis,nowzari2016epidemics,ferguson2020report,giordano2021modeling}. The overall control objective often depends on the model in question. For SIR-type models, where the disease eventually dies out, objectives include reducing the peak infection level (the so called ``flattening the curve'' concept)~\cite{di2020covid}, and limiting the total number of removed individuals~\cite{yi2022edge}. For SIS-type models, where the disease can become endemic, a typical objective is to eliminate the disease entirely~\cite{preciado2014optimal}, and if not possible, then suppress and reduce the level of endemic infections~\cite{wang2022sis_feedback}. Evidently, the control actions in the model should reflect real-world public health interventions. For instance, medical interventions (e.g. increase of medicines, healthcare staff) can be modelled by increasing recovery rates~\cite{liu2019analysis,preciado2014optimal,jafarizadeh2023optimal}, or vaccinating susceptible individuals~\cite{miller2007effective}. Non-pharmaceutical interventions such as wearing masks, physical distancing, or restricting population mobility, can be modelled by decreasing infection rates or, for networked models, also by removing nodes and/or edges~\cite{al2021long,gevertz2021novel,yi2022edge,preciado2014optimal}.

Our work focuses on the continuous-time deterministic networked SIS model~\cite{lajmanovich1976deterministic}, the limiting behaviour of which is characterised by a reproduction number $\mathcal{R}_0$ computed as a complex nonlinear function of the recovery and infection rates of the nodes (populations) in the network~\cite{ye2021_PH_TAC,shuai2013epidemic_lyapunov,mei2017epidemics_review}. Namely, if $\mathcal{R}_0 \leq 1$, then the network converges to the healthy disease-free state (the disease is eliminated from every node), whereas in contrast, for $\mathcal{R}_0 > 1$, the disease becomes endemic and infects a fraction of individuals in every population in the network. As a consequence, one reasonable problem is to focus on controlling networks with $\mathcal{R}_0 > 1$. 
A large body of literature on SIS-type models consider ``one-shot'' control approaches, whereby the network is modified via node/edge removal or static adjustment of recovery/infection rates, in order to reduce $\mathcal{R}_0$ as much as possible. In fact, where possible, one would aim to reduce $\mathcal{R}_0$ to being below $1$~\cite{preciado2014optimal,zhai2013optimization,wan2008designing,somers2020sis_optimisation,holme2002attack,miller2007effective,van2011decreasing,jafarizadeh2023optimal}. Often, a budget is set for allowable modifications, and network optimisation approaches are used to identify the optimal set of nodes/edges to remove and/or adjust. One key limitation is that ``one-shot'' approaches are non-dynamic, whereas real-world interventions may be dynamically updated as an outbreak unfolds. A second limitation is that for many works, both control design and implementation is centralised, requiring full information on the network, including all recovery and infection rates, which may be difficult to obtain (especially for novel diseases or outbreaks).


Another approach is to ``close-the-loop'' via state feedback methods
to dynamically adjust infection and/or recovery rates using information on infection numbers~\cite{ye2021_PH_TAC,liu2019analysis,wang2022sis_feedback}. While one can consider state feedback control for a single population~\cite{al2021long,di2020covid}, it is especially suited for networked populations because it can often be decentralised in both design and execution. Specifically, each node can make independent adjustments based on its own infected numbers, and little-to-no information about the overall network structure or state is required to design algorithm parameters or execute said algorithms. Both these features improve on the ``one-shot'' approaches discussed above. However, if the networked SIS model has $\mathcal{R}_0 > 1$, then to the best of the authors' knowledge, existing decentralised state feedback controllers can only mitigate the epidemic (reduce the level of endemic infections in each population) but cannot \textit{eliminate} the epidemic, i.e., drive the network to the healthy disease-free state~\cite{ye2021_PH_TAC,liu2019analysis,wang2022sis_feedback}. 

In this paper, we address the various limitations noted above by proposing a family of decentralised adaptive-gain control algorithms, which are able to eliminate the epidemic from every node in the networked SIS model. We consider control of the infection and recovery rates separately. The aforementioned state feedback approaches lower the infection rate or increase the recovery rate at a node as a monotonic function of the number of infected individuals at the node~\cite{ye2021_PH_TAC,liu2019analysis,wang2022sis_feedback}. In contrast, here, we adjust the rates via  multiplicative gains which evolve adaptively via a differential equation, utilising information about the current fraction of infected individuals; no information about the network, such as recovery or infection rates, is required to execute the algorithm. 

We first consider the full network control problem; either every node adaptively adjusts its infection rate, or adaptively adjusts its recovery rate. Our key theoretical result is to prove that for the proposed family of control algorithms, the network converges asymptotically to the healthy disease-free state. Importantly, we show that for the proposed adaptive algorithms, every gain converges to a \textit{positive} and \textit{finite} value. This ensures our algorithm is well-posed; nodes are not expected to completely isolate themselves (infection control gain converges to zero) or provide infinite medical resources (recovery control gain tends to infinity). No knowledge about the network is required for both algorithm design and algorithm execution. Auxiliary results establish bounds on the limiting gain values, properties of system trajectories, and an exponential convergence property for a subclass of the controllers. 

Then, we consider the partial network control problem, again for both infection and recovery rate control. Here, we identify a necessary and sufficient condition on the network (involving infection and recovery parameters as well as the network structure), such that there exists a proper subset of the nodes to which one can apply the adaptive-gain controllers and guarantee the same convergence outcome as in the full network control case. That is, the disease is eliminated from every node, and gains remain strictly positive. We then propose a centralised iterative algorithm to select a suitable subset of nodes to control, assuming the necessary and sufficient condition is satisfied. Complete knowledge of the network is required to check the necessary and sufficient condition and to run the iterative node selection algorithm. However, the actual adaptive control algorithm execution remains decentralised and requires no knowledge of the network structure or infection/recovery rate parameters. Simulations on both large-scale and smaller networks are provided to shed further light on the theoretical findings. A preliminary version of this work will appear in the 22nd IFAC World Congress~\cite{Walsh2023_IFAC}, covering only the full infection rate control problem, and not addressing the full recovery rate control or the partial network control.

The rest of the paper is structured as follows. Section~\ref{sec:model_problem} presents the SIS network model and motivates the adaptive-gain problem. Sections~\ref{sec:full_control} and \ref{sec:partial_control} deal with the full network control and partial network control problems, respectively. The paper is concluded in Section~\ref{sec:con}.

\subsection{Notation}
The $n$-column vectors of all ones and of all zeros are denoted by $\mathbf{1}_n$ and $\mathbf{0}_n$, respectively. The $n\times n$ identity matrix and the $m\times n$ zero matrix are given by $I_n$ and $\mathbf{0}_{m\times n}$, respectively. The $i$th component of a vector $a$ and the $(i,j)$-th entry of a matrix $A$ are given by $a_i$ and $a_{ij}$, respectively. For two vectors $a,b\in \mathbb{R}^n$, we write $a>b$ if $a_i>b_i$ for all $i$ and $a\geq b$ if $a_i\geq b_i$ for all $i$. A real matrix $A\in \mathbb{R}^{m\times n}$ is said to be nonnegative if all its entries are nonnegative, \textit{i.e.} $a_{ij}\geq 0$ for all $i,j$, and we write $A\geq \mathbf{0}_{m\times n}$. For a real square matrix $A\in \mathbb{R}^{n\times n}$ with spectrum $\sigma(A)$, we define $\rho(A):=\mathrm{max} \{|\lambda| : \lambda \in \sigma(A)\}$ and $s(A):=\mathrm{max}\{\mathrm{Re}(\lambda) : \lambda \in \sigma(A)\}$ as the spectral radius of $A$ and the spectral abscissa of $A$, respectively.

For a set $\Omega$, $\mathrm{Int}(\Omega)$ denotes its interior.
We define the $n$-dimensional closed unit hypercube as
\[\Xi_n:=\{x\in \mathbb{R}^n : \ 0\leq x_i\leq 1, \forall i=1, \hdots, n\}.\] 

\subsection{$\mathcal{L}^p$ Function Spaces}
Consider a function $f(t) : \mathbb{R}_{\geq 0} \to \mathbb R$ that is locally integrable. Given a fixed $p\in [1, \infty)$, we say that $f(t)$ belongs to the $\mathcal{L}^p$ space if $\int_0^{\infty} |f(s)|^pds < \infty$. We define the function $p$-norm as $ \Vert f(t) \Vert_{\mathcal{L}^p} = \left( \int_0^{\infty} |f(s)|^pds\right)^\frac{1}{p}$.
We say $f(t)$ belongs to $\mathcal{L}^\infty$ if and only if $\text{ess sup}_{t\geq 0} |f(t)| < \infty$, where ``ess sup'' denotes the \textit{essential supremum}. The function $\infty$-norm is given by $\Vert f(t) \Vert_{\mathcal{L}^\infty} = \text{ess sup}_{t\geq 0} |f(t)|$.

For a vector-valued function $f(t) : \mathbb{R}_{\geq 0} \to \mathbb R^n$, we say that $f(t)$ belongs to $\mathcal{L}^p_n$ space if $\int_0^{\infty} \Vert f(s)\Vert ^p ds < \infty$ where $\Vert \cdot \Vert$ is the Euclidean norm (note that any other vector norm could be used for the definition, due to equivalence of the norms). We define the function $p$-norm in $\mathbb R^n$ as $ \Vert f(t) \Vert_{\mathcal{L}^p_n} = \left( \int_0^{\infty} \Vert f(s)\Vert^pds\right)^\frac{1}{p}$.
We say $f(t)$ belongs to $\mathcal{L}^\infty_n$ if and only if $\text{ess sup}_{t\geq 0} \Vert f(t)\Vert < \infty$. The function $\infty$-norm in $\mathbb R^n$ is given by $\Vert f(t) \Vert_{\mathcal{L}^\infty_n} = \text{ess sup}_{t\geq 0} \Vert f(t)\Vert$.


\subsection{Graph Theory}
A directed graph is a triple $\mathcal{G}=(\mathcal{V},\mathcal{E}, B)$, where $\mathcal{V}=\{1, \hdots, n\}$ is the set of vertices (or nodes), $\mathcal{E}\subseteq \mathcal{V}\times \mathcal{V}$ is the set of edges and $B\geq \mathbf{0}_{n\times n}$ is the nonnegative weighted adjacency matrix which encodes $\mathcal{E}$ by the rule $(j,i)\in \mathcal{E}\iff b_{ij}>0$. A path from node $p_1\in \mathcal{V}$ to node $p_m\in \mathcal{V}$ is a sequence of edges of the form $(p_1,p_2),(p_2,p_3), ..., (p_{m-1},p_m)$, where each $p_i\in \mathcal{V}$ is distinct and $(p_i,p_{i+1})\in \mathcal{E}$ for all $i$. If such a path from node $p_1\in \mathcal{V}$ to node $p_m\in \mathcal{V}$ exists, we say that node $p_m$ is reachable from node $p_1$. A directed graph $\mathcal{G}$ is \textit{strongly connected} if for every pair of vertices $i,j\in \mathcal{V}$, $j$ is reachable from $i$, which is equivalent to its weighted adjacency matrix $B$ being irreducible~\cite{godsil2001algebraic}. A \textit{simple cycle} is a modification of a path in which the first and last nodes, $p_1$ and $p_m$, are identical; note that every edge in a simple cycle is distinct, and we do not consider a self-loop as a simple cycle.

\section{SIS Network Model and Problem Motivation}\label{sec:model_problem}
In this section, we introduce the SIS network model and motivate the adaptive-gain control problem.

\subsection{The Deterministic SIS Network Model}

The deterministic SIS network model is a classical model within mathematical epidemiology~\cite{nowzari2016epidemics,lajmanovich1976deterministic,mei2017epidemics_review}. We consider a network with $n\geq 2$ large, distinct, well-mixed populations\footnote{The notions of `large' and `well-mixed' have technical definitions, as discussed in~\cite{nowzari2016epidemics}. These two assumptions ensure that \eqref{eq:sisdynamics_vect} is a mean-field approximation of the stochastic model, the latter being a more accurate reflection of the true epidemic spreading process but significantly more challenging to analyse.} encoded by a graph $\mathcal{G} = (\mathcal{V}, \mathcal{E}, B)$. Each node $i\in \mathcal{V}$ denotes a population of fixed size comprising individuals which, as noted in the Introduction, belong to one of two mutually exclusive \textit{health compartments}: Susceptible and Infected. 
We let $x_i(t)\in [0,1]$ denote the proportion of Infecteds in population~$i$ at time $t \geq 0$ and thus $1-x_i(t)$ is the proportion of Susceptibles in population~$i$. The dynamics of state $x_i(t)$ are:
\begin{equation}\label{eq:sisdynamics_ind}
\dot{x}_i(t) = -d_ix_i(t) + \left(1-x_i(t)\right)\sum_{j=1}^n{b_{ij}x_j(t)}
\end{equation}
where the \textit{recovery parameter} $d_i>0$ is the rate of recovery in the $i$th population and the \textit{infection parameter} $b_{ij}\geq0$ denotes the rate at which the Infecteds of population $j$ transmit the disease to the Susceptibles of population $i$. 


Letting $x(t)=\left[x_1(t), \hdots, x_n(t)\right]^\top$, the network infection dynamics can be compactly expressed as
\begin{equation}\label{eq:sisdynamics_vect}
\dot{x}(t) = -Dx(t)+\left(I_n-X(t)\right)Bx(t),
\end{equation}
where $D=\mathrm{diag}(d_1, \hdots, d_n)$, $X(t) = \mathrm{diag}\left(x_1(t), \hdots, x_n(t)\right)$ are diagonal matrices, and $B\geq \mathbf{0}_{n \times n}$ is a nonnegative square matrix with $(i,j)$th entry $b_{ij}$. One can show that the system in \eqref{eq:sisdynamics_vect} is well-defined in the sense that if $x_i(0)\in [0,1]$ for any population $i$, then $x_i(t)\in [0,1]$ for all $t\geq0$. Indeed, we state this formally in the following result, with various proofs appearing in \cite{lajmanovich1976deterministic,mei2017epidemics_review,van2008virus}.

\begin{lemma}\label{lem:uncontsys_pos_inv}
Consider the system in \eqref{eq:sisdynamics_vect} and suppose that $x(0)\in \Xi_n$. Then $x(t)\in \Xi_n$ for all $t\geq 0$.
\end{lemma}
As it turns out, it will be convenient to consider the epidemic network dynamics from the perspective of the directed graph $\mathcal{G} = (\mathcal{V}, \mathcal{E}, B)$, being associated with the weighted adjacency matrix $B$. 
Directed edges allow description of heterogeneous transmission rates between populations, which can often occur. For instance, in the context of gonorrhoea transmission, population~$i$ and population~$j$ may represent a female- and male-only group, and we would not automatically expect $b_{ij}$ and $b_{ji}$ to be equal~\cite{lajmanovich1976deterministic,Yorke1978}. We make the following standing assumption in the paper, which is standard~\cite{nowzari2016epidemics,mei2017epidemics_review,lajmanovich1976deterministic,ye2021_PH_TAC}. 

\begin{assumption}\label{assm:strongly_connected}
The graph $\mathcal{G} = (\mathcal{V}, \mathcal{E}, B)$ is strongly connected, and the matrix $D$ is positive diagonal.
\end{assumption}

Strong connectivity ensures that a transmission pathway (possibly involving intermediate nodes) exists between any two pairs of nodes. It is equivalent to $B$ being irreducible, and is not especially restrictive in the epidemic modelling context (e.g. any undirected connected graph is strongly connected). Network SIS models in which the underlying graph is not strongly connected have only received limited attention in the literature~\cite{khanafer2016SIS_positivesystems}.

Theorem~\ref{thm:R0}, below, establishes the formula for the basic reproduction number $\mathcal{R}_0$, whose value uniquely determines the long-term presence of the disease on the network. Several different proofs can be found due to \cite{lajmanovich1976deterministic,van2008virus,ye2021_PH_TAC,mei2017epidemics_review}.

\begin{theorem}\label{thm:R0}
Consider the system in \eqref{eq:sisdynamics_vect} under Assumption~\ref{assm:strongly_connected}. Define $\mathcal{R}_0:=\rho(D^{-1}B)$. Then,
\begin{enumerate}
    \item If $\mathcal{R}_0\leq 1$, $\mathbf{0}_n$ is the unique equilibrium point of \eqref{eq:sisdynamics_vect} and $\lim_{t\to\infty}{x(t)}=\mathbf{0}_n$ for all $x(0)\in\Xi_n$. Convergence is exponentially fast only if $\mathcal{R}_0 < 1$.
    \item If $\mathcal{R}_0>1$, then in addition to $\mathbf{0}_n$, which is an unstable equilibrium point, there exists exactly one other equilibrium point $x^*\in \mathrm{Int}\left(\Xi_n\right)$ such that, for every $x(0)\in \Xi_n\backslash \mathbf{0}_n$, $\lim_{t\to\infty}{x(t)}=x^*$ exponentially fast.
\end{enumerate}
\end{theorem}

Since $x=\mathbf{0}_{n}$ denotes a network state wherein all populations are disease-free, we refer to $\mathbf{0}_{n}$ as the \textit{healthy equilibrium}. Similarly, for a network with $\mathcal{R}_0>1$, we call the additional equilibrium point $x^*\in \mathrm{Int}\left(\Xi_n\right)$ an \textit{endemic equilibrium} since it denotes a network state in which all populations contain some (nonzero) fraction of Infecteds.

\subsection{Problem Motivation}\label{ssec:motivation}

We discuss the motivations for an adaptive-gain approach using the single population SIS model, i.e. $n = 1$ and $x(t) \in [0, 1]$, as an exemplar. The dynamics are
\begin{equation}\label{eq:sis_single_uncontrolled}
    \dot{x}(t) = -d x(t) + \left(1-x(t)\right) b x(t),
\end{equation}
and an epidemic outbreak occurs if $\mathcal{R}_0 = b/d > 1$. 
Non-pharmaceutical interventions (NPIs), such as physical distancing and mobility restrictions, are used by policymakers to reduce the rate of disease transmission, and might be the only viable option for novel diseases that have no vaccines or medicines readily available. Reduction of the infection rate can be captured by adjusting \eqref{eq:sis_single_uncontrolled} to read:
\begin{equation}\label{eq:sis_single_control_infection}
    \dot{x}(t) = -d x(t) + (1-x(t))gbx(t),
\end{equation}
where $g \in [0,1]$ is a `control gain' that represents the effectiveness of NPIs in lowering the infection rate $b$. Adjustment of $b$ via a multiplicative term $g$ is a standard approach to represent NPIs, see e.g.~\cite{moore2021vaccination_NPI_COVID,wong2020modeling_COVID_NPI,tian2020investigation}. Thus, $gb$ can be considered the `controlled' infection rate, and policymakers will generally seek to design NPIs that ensure $gb/d \leq 1$, which is effectively what is done with the so-called ``one-shot'' methods, and this will guarantee the disease becomes extinct as $t\to\infty$. 

Intuitively, one must balance decreasing $g$ (i.e., increasing the strength of the NPIs) enough as to suppress the epidemic but not too much as to impose significant social-economic costs to the population and policymakers. In reality, policymakers will typically introduce control actions in several phases over time (perhaps progressively more severe restrictions) until the epidemic is suppressed. This can be represented by the gain $g(t)$ decreasing over time, and a na\"{i}ve approach is to continuously decrease $g(t)$ until it reaches $0$. Evidently, there is a motivation to consider adaptive approaches that employ state feedback, so that $g(t)$ need not be continuously decreasing, and indeed does not decrease by much more than is necessary to eliminate the disease. 

If we instead consider applying a control gain to the recovery rate, the dynamics become $\dot{x}(t) = -d g x(t) + (1-x(t))bx(t),$ with $g \in [1, \infty)$. Here, the gain may represent increased medical resources (personnel or medication) to allow faster recovery from disease, and such an approach would be more appropriate for diseases for which medical interventions are widely available. Thus, one can similarly envisage the desire to consider adaptive-gain control to adjust $g(t)$ as applied to the recovery rate, in response to an ongoing epidemic outbreak.

For convenience, we refer to \textit{infection rate control} and \textit{recovery rate control} as applying the gain $g$ as $gb$ and $gd$, respectively. While our exemplar considered a single population, in this paper, we will explore a series of infection and recovery rate networked control problems using decentralised adaptive-gain control. By \textit{decentralised}, we mean that each node independently executes a control law and requires only measurement of its own state $x_i(t)$. Moreover, consistently with the common assumption in adaptive control that some parameters are unknown so that adaptive control needs to be used,  our particular adaptive-gain controller will not require knowledge of the infection or healing rate parameters to execute. Our main results, presented in Section~\ref{sec:full_control}, first focus on \textit{full network control}, where every node executes a decentralised controller. In Section~\ref{sec:partial_control}, we consider \textit{partial network control}, where only a strict subset of the nodes are controlled.


In each of the problems considered in this paper, we begin by assuming that we have an SIS network whose state $x(t)$ converges from any $x(0)\in \Xi_n\backslash \mathbf{0}_n$ to an endemic equilibrium \textit{in the event that} no control is implemented. This is summarised by our second standing assumption:
\begin{assumption}\label{assm:endemic_R0}
There holds $\mathcal{R}_0 \triangleq \rho(D^{-1}B) > 1$. 
\end{assumption}
For each of the particular problems defined in the subsequent sections, the control objective is the same: we aim to design a class of adaptive-gain controllers that drive $x(t)$ to the healthy equilibrium $\vect 0_n$ from any $x(0)\in \Xi_n$, i.e. to eliminate the disease from the entire network. 


\section{Full Network Control}\label{sec:full_control}

In this section, we consider two distinct
problems concerning full network control of the SIS network, in which decentralised controllers are applied to each node to drive the network to the healthy state. We first study the infection parameter control problem, and then the recovery parameter control problem. 

\subsection{Infection Rate Control}\label{ssec:full_infection}

Based on Section~\ref{ssec:motivation}, we propose to model control of the infection parameters by applying to each node $i\in\mathcal{V}$ a gain $g_i(t)$ that reduces the infection rate $b_{ij}$ from every node $j$ that has an edge incoming to node $i$. Each gain $g_i:\mathbb{R}\to \mathbb{R}$ evolves according to the adaptive-gain control law
\begin{equation}\label{eq:adaptive_law}
\dot{g}_i(t) = -\phi_i(x_i(t))g_i(t),\quad g_i(0)=1,
\end{equation}
where $\phi_i:[0,1]\to \mathbb{R}$ is a function satisfying the properties listed in Assumption~\ref{ass:phi_properties}, below. 

\begin{assumption}[Properties of $\phi_i$]\label{ass:phi_properties}
For some positive integer $p \in \mathbb N_+$, there holds $\phi_i(x_i) = \alpha_i {x_i}^{p}$ with tuning parameter $\alpha_i > 0$, for every $i\in \mathcal{V}$.
\end{assumption}

Assumption~\ref{ass:phi_properties} implies that $\phi_i(x_i)$ is continuously differentiable on $[0,1]$, and $\phi_i(0) = 0$, $\phi_i(x_i)>0$ for all $x_i\in (0,1]$.  


Formally, the controlled node dynamics are 
\begin{subequations}\label{eq:cont_system_infection}
\begin{align}
\dot{x}_i(t) &= -d_ix_i(t) + \left(1-x_i(t)\right) g_i(t) \sum_{j=1}^n{b_{ij}x_j(t)} \\
\dot{g}_i(t) &= -\phi_i(x_i(t))g_i(t), \ \ g_i(0)=1.
\end{align}
\end{subequations}
Note that, for simplicity, we assume $g_i(0)= 1$, as this represents the general scenario where no controls are applied at the initial outbreak. However, every result in Section~\ref{ssec:full_infection} can be easily extended to allow $g_i(0) \in (0,1]$. Similarly to Section~\ref{ssec:motivation}, we can interpret the control gain $g_i(t)$ as being applied to all of population $i$: susceptible individuals in population~$i$ (the term $1-x_i(t)$) are being infected over the network by infectious individuals (the term $\sum b_{ij} x_j(t)$), with the total infection `force' adjusted by $g_i(t) \in [0,1]$ to read as $(1-x_i(t))g_i(t) \sum b_{ij} x_j(t)$.

\begin{remark}
In real-world applications, interventions are introduced in phases, and hence $g_i(t)$ would be implemented as a piecewise constant control gain, rather than updated continuously as in \eqref{eq:cont_system_infection}.  Nonetheless, and as we show in the sequel, the study of \eqref{eq:cont_system_infection} provides important insights into the success of adaptive-gain approaches to epidemic control, and our simulations (see Section~\ref{ssec:sim_fullnetwork}) confirm that the adaptive-gain approach remains effective with piecewise constant updating of $g(t)$. The change to piecewise constant gains can be regarded as a type of iterative identification and control strategy, see~\cite{albertos2012iterative}, which is commonly used as a variation to more standard adaptive control. \hfill $\triangle$
\end{remark}

Let us define $g(t) = \left[g_1(t), \hdots, g_n(t)\right]^\top \in \Xi_n$ and the diagonal matrices $\Phi(x) = \mathrm{diag}(\phi_1(x_1), \hdots, \phi_n(x_n))$ and $G = \mathrm{diag}\left(g_1,\hdots,g_n\right)$.
By defining $\xi(t)=[x(t)^\top, \ g(t)^\top]^\top\in \mathbb R^{2n}$, we can compactly express the network dynamics as
\begin{equation}\label{eq:simp_contsystotal}
    \dot \xi(t)=f(\xi(t))
\end{equation}
where $f:\Xi_n\times\Xi_n\to\mathbb{R}^n\times\mathbb{R}^n$ is a map defined by
\begin{equation}\label{eq:f_infection}
    f(\xi(t))=
\begin{bmatrix}
-Dx(t)+\left(I_n-X(t)\right)G(t)Bx(t) \\
-\Phi(x(t))g(t)
\end{bmatrix}
\end{equation}


The first problem of this paper can now be stated.

\begin{problem}\label{prob:totalinfcontrol}
Consider the system in \eqref{eq:simp_contsystotal} under Assumptions~\ref{assm:strongly_connected}, \ref{assm:endemic_R0} and \ref{ass:phi_properties}. Show that a decentralised controller gain $g_i(t)$ subject to the adaptive control law in \eqref{eq:adaptive_law} for each $i\in \mathcal{V}$ ensures that i) $\lim_{t\to\infty}{x(t)}=\mathbf{0}_n$ for any $x(0)\in \Xi_n$, and ii) $\lim_{t\to\infty} g_i(t) > 0$ for all $i\in\mathcal{V}$.
\end{problem}

To begin, we show that for the system in \eqref{eq:simp_contsystotal}, the set $\Xi_n\times~ \Xi_n$ is positively invariant. That is, if $[x(0)^\top, g(0)^\top]^\top\in \Xi_n\times \Xi_n$, then $[x(t)^\top, g(t)^\top]^\top\in \Xi_n\times \Xi_n$ for all $t\geq 0$. This ensures the model and control algorithm are well-defined within the epidemic context. 

\begin{lemma}\label{lem:cont_sys_pos_invariance}
Consider the system in \eqref{eq:simp_contsystotal} and suppose that $\xi(0)\in\Xi_n\times \Xi_n$. Then $\xi(t)\in \Xi_n\times \Xi_n$ for all $t\geq 0$. 
\end{lemma}
\begin{proof}
Evidently, the map $f$ from \eqref{eq:simp_contsystotal} is Lipschitz over its compact domain $\Xi_n\times \Xi_n$, and solutions of the differential equation in \eqref{eq:simp_contsystotal} are unique. As a consequence, we can use Nagumo's Theorem~\cite{blanchini1999set_invariance} to establish the invariance of $\Xi_n \times \Xi_n$, as follows.

Consider an arbitrary $i\in \mathcal{V}$ at some finite $t\geq 0$, and suppose that $x_j(t) \in [0,1]$ for all $j\neq i$. It is immediate from the conditions imposed on $\phi_i$ that $g_i(t) \in [0,1]$ for all $t$. Now consider $\dot{x}_i(t)$ for $x_i(t) = 1$. Clearly, $\dot{x}_i(t) = -d_i < 0$. Conversely, for $x_i(t) = 0$, we have $\dot{x}_i(t) = g_i(t)\sum_{j=1}^n b_{ij} x_j(t) \geq 0$. Since this holds for any $i$, it follows from Nagumo's Theorem that $\xi(t) \in \Xi_n \times \Xi_n$ for all $t\geq 0$~\cite{blanchini1999set_invariance}.
\end{proof}

Notice that $\xi=[\mathbf{0}_n^\top, \bar g^\top]^\top$ is an equilibrium of the system in \eqref{eq:simp_contsystotal} for any controller input gain $\bar g\in \Xi_n$. The set $\Omega = \mathbf{0}_n~\times~\Xi_n$ therefore consists of equilibrium points in which the virus is extinct at every node. If we can show that every trajectory $\xi(t)$ of the system starting in $\Xi_n\times \vect 1_n$ converges to a point in the set $\Omega$, then we have solved Problem~\ref{prob:totalinfcontrol}. Indeed, this is what we demonstrate by the following Theorem~\ref{thm:full_infection}.

\begin{theorem}\label{thm:full_infection}
Consider the system in \eqref{eq:simp_contsystotal} under Assumptions~\ref{assm:strongly_connected}, \ref{assm:endemic_R0} and \ref{ass:phi_properties}. 
Then, for all $\xi(0)\in \Xi_n\times \vect 1_n$, there holds $\lim_{t\to\infty} x(t) = \mathbf{0}_n$  and $\lim_{t\to\infty} g(t) = \bar g > \mathbf{0}_n$.
\end{theorem}
\begin{proof}
The solution to the differential equation $\dot{g}_i$ in \eqref{eq:cont_system_infection} yields, for each $i$,
\begin{equation}\label{eq:g_solution}
    g_i(t) = g_i(0)e^{-\int_0^t \phi_i(x_i(s))ds}.
\end{equation}
Since $\phi_i(x_i(t))\geq 0$ for each $i$ and for all $t\geq 0$ by Lemma~\ref{lem:cont_sys_pos_invariance}, the integral $\int_0^t{\phi_i(x_i(s))}ds$ is monotone non-decreasing and thus approaches a limit as $t\to\infty$, and this limit is either finite or infinite. This implies that $\lim_{t\to\infty} g_i(t)$ is either strictly positive or zero, respectively. With inessential reordering of the node indices if necessary, suppose that for $i \in \mathcal{V}_{I} \triangleq \{1, 2, \dots, k\}$, $\int_0^\infty{\phi_i(x_i(s))}ds$ is infinite and for $i\in \mathcal{V}_{F} \triangleq \{k+1, k+2, \dots, n\}$, $\int_0^\infty{\phi_i(x_i(s))}ds$ is finite. We allow for the moment both extreme cases, i.e., where all integrals are finite and all integrals are infinite, where $k=0$ and $k=n$, respectively.

Our proof will first establish convergence of $x_i(t)$ to $0$ for $i\in \mathcal{V}_{F}$. Then, we prove by contradiction that for $i\in\mathcal{V}_I$, $\lim_{t\to\infty} g_i(t) > 0$, which implies that $\mathcal{V}_I$ is empty. 



To begin, as $t\to\infty$ we have that $g_i(t) \to \bar g_i$ where $\bar g_i>0$ for $i \in \mathcal{V}_F$. The uniform continuity of $\phi_i(x_i(t))$, the finiteness of $\int_{0}^{\infty}\phi_i(x_i(s)ds$ and Barbalat's Lemma, \cite[see p. 323]{khalil2002nonlinear} ensure that $\phi_i(x_i(t)) \to 0$ as $t\to\infty$, which implies that $\lim_{t\to\infty } x_i(t) = 0$. By hypothesis, for $i\in \mathcal V_F$, there holds  $\int_0^\infty{\phi_i(x_i(s))}ds = \alpha_i \int_{0}^\infty \big(x_i(s)\big)^{p}ds < \infty$, where $p \in \mathbb N_+$ defines the adaptive gain algorithm as given in Assumption~\ref{ass:phi_properties}. In other words, $x_i(t) \in \mathcal{L}^p$ for all $i\in \mathcal{V}_F$. 

We now turn to $i\in \mathcal V_I$, recalling that by definition $\lim_{t\to\infty} g_i(t) = 0$ for all $i\in\mathcal{V}_I$. We claim that $\mathcal{V}_I$ is empty. To obtain a contradiction, suppose that $\mathcal{V}_I$ is not empty, i.e. $k \neq 0$. Let us define, $\tilde x = [x_1, \hdots, x_k]^\top \in \mathbb{R}^k$ and $\hat x = [x_{k+1},\hdots, x_n]^\top \in \mathbb{R}^{n-k}$. From \eqref{eq:cont_system_infection}, the differential equation for $\Tilde{x}(t)$ is:
\begin{align}\label{eq:prob1_I_sisdynamics_vect}
    \dot{\Tilde{x}}(t)& =-\Tilde{D}\Tilde{x}(t)+\Big(I_k-\Tilde{X}(t)\Big)\Tilde{G}(t)\Tilde{B}\Tilde{x}(t) + w(t),
\end{align}
where $\tilde{D} = \diag(d_1, \dots, d_k)$, $\tilde{G}(t) = \diag(g_1(t), \dots, g_k(t))$, $\tilde{X}(t) = \diag(x_1(t), \dots, x_k(t))$,  $\tilde B \in \mathbb{R}^{k\times k}$ has $(i,j)$-th entry equal to $b_{ij}$, and $w(t) =(I_k-\tilde X(t))\Tilde{G}(t)\Hat{B}\Hat{x}(t)$. We can consider $w(t)$ to be an input signal in \eqref{eq:prob1_I_sisdynamics_vect}, where $\hat{B}\in\mathbb{R}^{k\times (n-k)}$ has $(i,j)$-th entry equal to $b_{rq}$ for $r = i$  and $q = k + j$.

Since $\lim_{t\to\infty}g_i(t) = 0$ for all $i\in\mathcal V_I=\{1,2,\dots,k\}$, it follows that for any $\epsilon > 0$, there exists some $\tau_i\geq 0$ such that $0\leq g_i(t)\leq \epsilon$ whenever $t\geq \tau_i$. Then for $t\geq \tau:=\max_{i=1, \dots, k}{\tau_i}$, it holds that $\Tilde{G}(t)\leq \epsilon I_k$. For $t\geq\tau$, \eqref{eq:prob1_I_sisdynamics_vect} thus evaluates to be
\begin{align}\label{eq:prob1_ineq}
    \dot{\Tilde{x}}(t) &= -\Tilde{D}\Tilde{x}(t)+\Big(I_k-\Tilde{X}(t)\Big)\Tilde{G}(t)\Tilde{B}\Tilde{x}(t) + w(t) \nonumber \\
    & \leq -\Tilde{D}\Tilde{x}(t)+\Tilde{G}(t)\Tilde{B}\Tilde{x}(t) + w(t) \nonumber \\
    & = \Big(-\Tilde{D}+\Tilde{G}(t)\Tilde{B}\Big)\Tilde{x}(t) + w(t) \nonumber \\
    & \leq \Big(-\Tilde{D}+\epsilon \Tilde{B}\Big)\Tilde{x}(t) + w(t).
\end{align}
The first inequality is obtained because $\mathbf{0}_{k\times k}\leq (I_k-\Tilde{X}(t)) \leq I_k$, while the second inequality is due to the fact that $\Tilde{G}(t)\leq \epsilon I_k$. By the Gershgorin Circle Theorem \cite[Theorem~6.1.1]{horn2012matrixbook}, we can take $\epsilon$ sufficiently small such that $A_\epsilon = -\Tilde{D} + \epsilon \Tilde{B}$ is Hurwitz. Assume such a choice of $\epsilon$ has been taken. Now, consider the system 
\begin{equation}\label{eq:positive_linear_system}
    \dot{\Tilde{y}}(t) = A_\epsilon \tilde y(t) + w(t),
\end{equation}
with $\tilde y(0)$ selected such that $\tilde y(\tau) = \tilde x(\tau)$. 

A sketch of our subsequent arguments is as follows. First, we will show that $w(t) \in \mathcal{L}^p_k$, which immediately implies that $\tilde y(t) \in \mathcal{L}^p_k$ and $\tilde y(t)\to\vect 0_{k}$. Then, we will show that $\tilde x(t) \leq \tilde y(t)$ for all $t$, which leads to the theorem result.

Recall that $(I_k - \tilde X(t))\tilde G(t)$ is a diagonal nonnegative matrix with diagonal entries less than or equal to 1. Recall further that the entries of $\hat x(t)$ are $x_i(t)$ for $i\in \mathcal{V}_F$, which implies that each entry of $\hat x$ is in $\mathcal{L}^p$ and $\lim_{t\to\infty} \hat x(t) = \vect 0_{n-k}$. It follows that $\vect 0_{k\times k} \leq \big(I_k-\Tilde{X}(t)\big)\Tilde{G}(t)\Hat{B}\Hat{x}(t)\leq \Hat{B}\Hat{x}(t)$. Clearly then, $\Hat{B}\Hat{x}(t) \in \mathcal{L}^p_k$, from which we immediately conclude that $w(t) \in \mathcal{L}^p_k$. Standard linear systems theory establishes that $\tilde y(t) \in \mathcal{L}^p_k$ and $\dot{\tilde y}(t) \in \mathcal{L}^p_k$~\cite[Theorem~9, pg.~59]{desoer2009feedback}.

We obtain from \eqref{eq:prob1_ineq} that, for $t\geq \tau$, 
\begin{equation*}
    \dot{\Tilde{x}}(t) - A_\epsilon \Tilde{x}(t) - w(t) \leq \dot{\Tilde{y}}(t) - A_\epsilon\Tilde{y}(t) - w(t).
\end{equation*}
Since $A_\epsilon$ has all diagonal entries nonnegative, \eqref{eq:positive_linear_system} satisfies condition Q of \cite{walter1971ordinary}, and so the main theorem of \cite{walter1971ordinary} establishes that $\Tilde{x}(t) \leq \Tilde{y}(t)$ for all $t\geq \tau$. 
We established above that $\Tilde{y}(t)\in\mathcal{L}^p_k$, while there holds $\tilde x(t) \geq \vect 0_n$ for all $t$ due to Lemma~\ref{lem:cont_sys_pos_invariance}. It follows that $\Tilde{x}(t)\in\mathcal{L}^p_k$.
However, this implies that for every $i \in\mathcal V_I$, $\int_0^t{\phi(x_i(s))} ds = \alpha_i \int_0^t x_i^p ds$ converges to a finite value as $t\to\infty$, and by \eqref{eq:g_solution}, there exists $\bar g_i>0$ such that $g_i(t)\to\bar g_i$ as $t\to\infty$: a contradiction. Thus, we must have $k=0$, and hence $\mathcal{V}_I$ is empty. It follows that $\mathcal{V}_F = \mathcal{V}$ and from our earlier analysis, we have that $\lim_{t\to\infty} x(t) = \vect 0_n$ and $\lim_{t\to\infty} g(t) = \bar g > \vect 0_n$.
\end{proof}

\begin{remark}[The role of parameter $p$]
The current Theorem~\ref{thm:full_infection} requires that all nodes use the same parameter $p$, i.e., our result does not allow for a mixture of different growth rates with $x_i$ for the gain functions: $\phi_i =\alpha_ix_i^{p_i}$ where $p_i \in \mathbb N_+$ and $\exists j,i$ with $j\neq i$ such that $p_i \neq p_j$. Different $\alpha_i$ are of course allowed. We first note that with heterogeneous $p_i \in \mathbb N_+$, it is straightforward to show that $\lim_{t\to\infty} x(t)  = \vect 0_n$; the challenge lies in proving $\lim_{t\to\infty} g(t) > \vect 0_n$, i.e. proving every gain $g_i(t)$ is strictly positive in the limit. Having $\lim_{t\to\infty} g_i(t)  = 0$ implies the node is eventually totally isolated with severe NPIs that restrict all mobility in the population, which from a practical point of view is extremely costly or even impossible.

The value of $p$ also plays a role in the speed of adaptation. More specifically, notice that for $x\in(0,1)$ and $a,b\in \mathbb N$ with $a>b$, we have $x^a < x^b$. Hence, the larger the $p$ value in \eqref{eq:adaptive_law}, the more slowly the gain adapts; intuitively, this means the limiting gain is larger (so the level of intervention is less strict) but the convergence to the healthy equilibrium is slower. Indeed, Proposition~\ref{prop:final_reproduction} establishes exponential convergence for $p = 1$, while simulations in Section~\ref{ssec:sim_fullnetwork} suggest that for $p=2$, convergence can occur at a rate of $1/t$, i.e., slower than exponential.
\hfill $\triangle$

\end{remark}

We conclude by providing several supplementary results that shed light on the limiting control gains and reproduction number of the controlled network. 
First, we recall for future use results on Metzler and $M$-matrices, and center manifolds of dynamical systems.

Let $A$ be a square matrix. We say that $A$ is a Metzler matrix if all off-diagonal entries are nonnegative. For an irreducible Metzler $A$, and by an extension of the Perron-Frobenius theorem \cite{horn1994topics_matrix}, $s(A)$ is a simple eigenvalue and the only eigenvalue with this real part. A corresponding eigenvector of $s(A)$ can be taken to have all positive entries, while no eigenvector corresponding to any other eigenvalue has this property. We say that $A$ is an $M$-matrix if $-A$ is Metzler and all eigenvalues of $A$ have positive real parts except for any at the origin. More specifically, $A$ is a singular or nonsingular $M$-matrix if it has at least one eigenvalue at the origin with all other eigenvalues having strictly positive real parts, or if its eigenvalues have strictly positive real parts, respectively~\cite{horn1994topics_matrix}. Further key properties, detailed in \cite[Theorem ~2.1]{varga2009matrix_book} and \cite[Theorem 2.3 and Theorem 4.6]{berman1979nonnegative_matrices}, are:
\begin{enumerate}
    \item
 For a (singular) $M$-matrix $F$, and any positive diagonal $D$, $DF$ is also a (singular) $M$-matrix.
 \item For an irreducible nonnegative matrix $B$ and positive diagonal matrix $D$, there holds i) $s(-D+B) > 0 \Leftrightarrow \rho(D^{-1}B) > 1$, ii) $s(-D+B) = 0 \Leftrightarrow \rho(D^{-1}B) = 1$ and iii) $s(-D+B) < 0 \Leftrightarrow \rho(D^{-1}B) < 1$.
 \end{enumerate}

Consider an autonomous dynamical system
\begin{equation}\label{eq:general_nl_system}
    \dot{x}(t) =  f(x(t)),
\end{equation}
and for an equilibrium point $\bar x$, let $J(\bar x) = \frac{\partial f}{\partial x}|_{\bar x}$ be the Jacobian matrix of $f$ evaluated at $\bar x$. Let $J(\bar x)$ have $p,q$ and $r$ eigenvalues with negative, positive, and zero real parts, respectively. The theory of stable, unstable, and center manifolds, states that for the eigenvalues with negative, positive and zero real parts, there are associated local invariant stable, unstable and center manifolds of the system in \eqref{eq:general_nl_system}, respectively~\cite[Theorem~3.2.1]{wiggins2003introduction}. These three manifolds are, crucially, tangent to the corresponding subspaces spanned by the eigenvectors of $J(\bar x)$ associated with the three sets of eigenvalues with negative, positive, and zero real parts. On the local unstable manifold, trajectories $x(t)$ move away from $\bar x$ at an exponential rate.



The following proposition provides an upper bound on~$\bar g_i$.

\begin{proposition}[Limiting gain upper bound]\label{lem:gaininequality}
Under the hypothesis of Theorem~\ref{thm:full_infection}, there holds
\begin{equation}
    \lim_{t\to\infty}g_i(t)\leq e^{-\frac{\alpha_ix^p_i(0)}{p d_i}},
\end{equation}
where $p \in \mathbb N_+$ defines the function $\phi_i = \alpha_i x_i^p$.
\end{proposition}
\begin{proof}
The differential equation for $x_i$ implies that $\dot x_i\geq -d_ix_i(t)$, from which it follows that $x_i(t)\geq \exp(-d_it)x_i(0)$. This yields $\big(x_i(t)\big)^p\geq \exp(-pd_it)\big(x_i(0)\big)^p$. We then obtain $\int_0^{\infty}(x_i(t))^p  dt\geq x_i^p(0)/pd_i,$
and then one can obtain $
\exp[-\alpha_i\int_0^{\infty}(x_i(t))^p dt]\leq \exp[-\frac{\alpha_ix_i^p(0)}{pd_i}]
$.
The left side is precisely $\lim_{t\to\infty}g_i(t)$ and the inequality is established. 
\end{proof}


\begin{proposition}[Properties of trajectory $x(t)$]\label{prop:positive_finite_t}
Suppose that $x(0) \geq \vect 0_n$ and there exists $i\in\mathcal{V}$ such that $x_i(0) > 0$. Then, under the hypothesis of Theorem~\ref{thm:full_infection}, there holds \mbox{$\vect 0_n < x(\tau) < \vect 1_n$} for any finite $\tau > 0$.
\end{proposition}

Note that under the hypothesis of Theorem~\ref{thm:full_infection}, $g_i(t) \geq \bar g_i > 0$ for all $i$, where $\bar g_i$ is a constant. We omit the proof of Proposition~\ref{prop:positive_finite_t} for brevity, since it follows the approaches used to establish a similar result for the standard SIS network models, see e.g.~\cite[Lemma~3.2]{lajmanovich1976deterministic}. Proposition~\ref{prop:positive_finite_t} states that the disease is present in all nodes at any finite time; it is not possible for the disease to only exist in a proper subset of the nodes. Note the disease may or may not be eliminated from all nodes as $t\to\infty$; under the hypotheses of Theorem~\ref{thm:full_infection}, disease elimination is assured for all nodes as $t\to\infty$. 
We can define the quantity $\mathcal{R}_t = \rho(D^{-1}G(t)B)$ as the `reproduction number' if the $g_i$ were frozen at time $t$, and Proposition~\ref{prop:positive_finite_t} then leads to the following result.

\begin{proposition}[Limiting reproduction number]\label{prop:final_reproduction}
Under the hypothesis of Theorem~\ref{thm:full_infection}, a limiting reproduction number, $\mathcal{R}_\infty = \rho(D^{-1}\bar GB)$ exists and satisfies $\mathcal{R}_\infty \leq 1$, where $\bar G = \diag(\bar g_1, \hdots, \bar g_n)$ is the positive diagonal matrix of the limiting control gains. In addition, $\mathcal R_{\infty}<1$ if $p = 1$.
\end{proposition}
\begin{proof}
From Theorem~\ref{thm:full_infection}, the trajectory $\xi(t)$ of \eqref{eq:simp_contsystotal} approaches $\bar \xi\triangleq \xi(\infty) =[\vect 0_n^\top, \bar g^\top]^\top$ asymptotically. The differential equation is time-invariant, and there is a well-defined Jacobian matrix $J(\bar \xi)$ at the equilibrium point $\bar \xi$. Using \eqref{eq:f_infection}, we can compute 
\begin{equation}
J(\bar \xi)=\begin{bmatrix}
-D+\bar GB&{\bf{0}}_{n\times n}\\
J_{21}&{\bf{0}}_{n\times n}
\end{bmatrix},
\end{equation}
where $J_{21}={\rm{diag}}(\alpha_i\bar g_i)$ if $p=1$ and $J_{21}={\bf 0}_{n\times n}$ if $p\geq2$. For convenience, let $Q\triangleq -D+\bar GB$.

Theorem~\ref{thm:full_infection} establishes that all trajectories of the nonlinear system beginning in $\Xi_n \times \Xi_n$ are known to be stable. Therefore, $\mathbb R^n_{\geq 0} \times \mathbb R^n_{\geq 0}$ cannot be part of an unstable manifold of \eqref{eq:simp_contsystotal} at the equilibrium $\bar \xi$. 
According to the theory of stable, unstable and center manifolds, there can then be no eigenvector $[u^\top, v^\top]^\top$ of $J$, with $u\in\mathbb{R}^n_{\geq 0}$, that corresponds to an eigenvalue of $J$ with positive real part (see \eqref{eq:general_nl_system}). 

We claim this implies $s(Q)\leq 0$. To show this by contradiction, suppose that $s(Q)>0$, and let an associated eigenvector be $u_Q$, which is known to be a positive vector. Define $v=(s(Q))^{-1}J_{21}u_Q \geq \vect 0_n$, and observe that $[u_Q^{\top},v^{\top}]^{\top}$, which is nonzero, is then a nonnegative eigenvector of $J$ with associated positive eigenvalue $s(Q)$. This is a contradiction. Last, note that $s(Q)\leq 0 \Leftrightarrow \rho(D^{-1}\bar GB)\leq 1$, as required, see 
above \eqref{eq:general_nl_system}.

To prove the final statement of the proposition, suppose now that $p=1$ and to establish a contradiction, assume that $\mathcal R_{\infty}=1$, or equivalently, that $s(Q)=0$. Let $\omega^{\top}>\vect 0_n^{\top}$ be the left eigenvector associated with the eigenvalue $s(Q)$, normalised to unit length for convenience. Let $\gamma>0$ be any  constant satisfying $\gamma>(b_{ij}/\omega_j)$ for all $i,j$. Note that as a consequence
\begin{equation}\label{eq:Bineq}
B<\gamma{\bf{1}_n}\omega^{\top}
\end{equation}
Now observe from \eqref{eq:cont_system_infection} that 
\[
\dot x=-Dx+G(t)Bx-G(t)XBx\geq(-D+\bar GB)x-G(t)XBx,
\]
where we have made use of the fact that the diagonal entries of $G(t)$ are monotone decreasing. Next, premultiply by $\omega^{\top}$, noting that $\omega^{\top}(-D+\bar GB)=\vect 0_n^{\top}$ by assumption. There results $\frac{d}{dt}(\omega^{\top} x)\geq -\omega^{\top}G(t)XBx \geq -\omega^{\top}XBx$, due to the fact that $G(t)\leq G(0)=I_n$. Using \eqref{eq:Bineq}, we further obtain $\frac{d}{dt}(\omega^{\top} x)\geq -\gamma\omega^{\top}X{\bf{1}}_n\omega^{\top}x$, which is equivalent to
\[\frac{d}{dt}(\omega^{\top} x)\geq-\gamma(\omega^{\top}x)^2.\]
This inequality implies that if $y=(\omega^{\top}x)^{-1}$, then $\dot y\leq \gamma$, so that $y$ grows no faster than linearly in $t$. It follows that $\omega^{\top}x$ decays no faster than at a rate $1/t$. Hence at least one entry of $x(t)$ also has this decay rate property, which we index as $x_i(t)$ for convenience. From \eqref{eq:cont_system_infection} and the fact that $\phi_i(x_i)=\alpha_ix_i$, we see that the corresponding entry of $g(t)$, namely $g_i(t)$, will then tend to $0$, which is a contradiction to what was proved in Theorem~\ref{thm:full_infection}. 
\end{proof}


Evidently, the disease is eliminated at an exponentially fast rate if $p=1$. Interestingly, simulations presented below in Section~\ref{ssec:sim_fullnetwork} for $p > 1$ appear to suggest that $\mathcal{R}_\infty = 1$, and convergence does not occur to $x = \vect 0_n$ at an exponentially fast rate. This suggests the choice of $p$ can play a significant role in the controlled dynamics, and deserves attention in future research.

\subsection{Recovery Rate Control}\label{ssec:full_recovery}

Similarly as with the control of infection rates, we propose to control the recovery rates by adjustment via an adaptive gain $g_i:\mathbb{R}\to \mathbb{R}$ that obeys
the control law
\begin{equation}\label{eq:adaptive_law_recov}
\dot{g}_i(t) = \phi_i(x_i(t)), \quad g_i(0) = 1,
\end{equation}
where $\phi_i:[0,1]\to \mathbb{R}$ is a function satisfying the properties listed in Assumption~\ref{ass:phi_properties}.

The control gain is implemented by replacing $d_{i}$ in \eqref{eq:sisdynamics_ind} with the expression $d_{i}g_i(t)$, where $d_{i}>0$ denotes the \textit{base recovery rate} against the disease for population~$i$, i.e., the recovery rate against the disease without control intervention. 
In total, the controlled dynamics at node~$i\in\mathcal{V}$ is given by
\begin{subequations}\label{eq:cont_system_recovery}
\begin{align}
\dot{x}_i(t) &= -g_i(t)d_ix_i(t) + \left(1-x_i(t)\right) \sum_{j=1}^n{b_{ij}x_j(t)} \\
\dot{g}_i(t) &= \phi_i(x_i(t)), \ \ g_i(0)=1.
\end{align}
\end{subequations}
Note that we are assuming $g_i(0)=1$, which reflects the situation wherein controls are not implemented at the beginning of an outbreak. The results reported in Section~\ref{ssec:full_recovery} easily extend to allow for $g_i(0) \in [1,\infty)$.
As will be demonstrated in the sequel, $g_i(t)$ in particular, and thus the whole system, does not exhibit finite time escape and thus the value of $g_i(t)$ is finite for all finite $t$.


With the same definitions of $g(t)$, $G(t)$, and $\Phi(x(t))$ as given below \eqref{eq:cont_system_infection}, we can write the controlled networked system dynamics as
\begin{equation}\label{eq:simp_cont_recovery}
    \dot \zeta(t) = h(\zeta(t)),
\end{equation}
but now considering a different state space vector of $\zeta(t)=[x(t)^\top, g(t)^\top]\in\Xi_n\times [1,\infty)^n$, and $h:\Xi_n\times [1,\infty)^n\to\mathbb{R}^n\times \mathbb{R}^n$ defined by,
\begin{equation}\label{eq:h_recovery}
    h(\zeta(t))= 
\begin{bmatrix}
-DG(t)x(t)+\left(I_n-X(t)\right)Bx(t) \\
\Phi(x(t))
\end{bmatrix}
\end{equation}

The second problem is now stated as follows.

\begin{problem}\label{prob:totalreccont}
Consider the system in \eqref{eq:simp_cont_recovery} under Assumptions~\ref{assm:strongly_connected}, \ref{assm:endemic_R0} and \ref{ass:phi_properties}. Show that a decentralised controller gain $g_i(t)$ subject to the adaptive control law in \eqref{eq:adaptive_law_recov} for each $i\in \mathcal{V}$ ensures that i) $\lim_{t\to\infty}{x(t)}=\mathbf{0}_n$ for any $x(0)\in \Xi_n$, and ii) $\lim_{t\to\infty} g_i(t) < \infty$ for all $i\in\mathcal{V}$.
\end{problem}


We firstly show that there is no possibility of finite escape time for the control inputs $g_i(t)\in[1,\infty)$, to ensure that solutions to \eqref{eq:simp_cont_recovery} are well-defined for all $t\geq 0$. 

\begin{lemma}\label{lem:finite_esc_time_prob3}
Consider the system in \eqref{eq:cont_system_recovery} under Assumptions~\ref{assm:strongly_connected}, \ref{assm:endemic_R0} and \ref{ass:phi_properties}. Then, for all $i\in\mathcal{V}$, it holds that $g_i(t)\leq(M_it)$ where $M_i:=\max_{x_i\in [0,1]}{\phi_i(x_i)}$.
\end{lemma}
\begin{proof}
The solution to the differential equation in \eqref{eq:cont_system_recovery} is $g_i(t) = g_i(0) + \int_0^t \phi_i(x_i(s))ds$. Since $\phi:[0,1]\to\mathbb{R}$ is continuous over a compact domain $[0,1]$, it attains some maximum $M_i:=\max{\phi_i(x_i)}$ by the Extreme Value Theorem. Therefore $M_i\geq \phi_i(x_i)$ for all $x_i\in[0,1]$ and so, since $x_i(t)\in[0,1]$ for all $t\geq0$ by Lemma \ref{lem:totalrec_posinv}, we have $M_i\geq \phi_i(x_i(t))$ for all $t\geq 0$. It then follows that $ g_i(t) = \int_0^t\phi_i(x_i(s))ds \leq\int_0^tM_ids = M_it$,
which delivers the claim of the lemma.
\end{proof}

A much tighter bound on $g_i(t)$ will be derived in the sequel; the weaker and more easily derived bound here simply suffices to demonstrate absence of an escape time. We next demonstrate that for the dynamical system \eqref{eq:simp_cont_recovery} the set $\Xi_n\times [1,\infty)^n$ is positively invariant. Together with Lemma~\ref{lem:finite_esc_time_prob3}, this result establishes that the controlled system of interest retains a meaningful interpretation within the epidemiological context.

\begin{lemma}\label{lem:totalrec_posinv}
Consider the system in \eqref{eq:cont_system_recovery} under Assumptions~\ref{assm:strongly_connected}, \ref{assm:endemic_R0} and \ref{ass:phi_properties}. Then we have $x_i(t)\in[0,1]$ and $g_i(t)\in[1,\infty)$ for all $i\in\mathcal{V}$ and for all $t\geq0$ if $x_i(0)\in[0,1]$ for all $i\in\mathcal{V}$. 
\end{lemma}
\begin{proof}
Because $f\in\mathcal{C}^1$ and the domain $\Xi_n\times [1,\infty)^n$ is closed as a product of closed intervals, Nagumo's Theorem~\cite{blanchini1999set_invariance} can be applied.
The remainder of the proof follows similarly as with Lemma~\ref{lem:cont_sys_pos_invariance}, and is thus omitted.
\end{proof}

We now provide the main result addressing Problem~\ref{prob:totalreccont}.

\begin{theorem}\label{thm:full_recovery}
Consider the system in \eqref{eq:cont_system_recovery} under Assumptions~\ref{assm:strongly_connected}, \ref{assm:endemic_R0} and \ref{ass:phi_properties}. 
Then, for all $\zeta(0)\in \Xi_n\times \vect 1_n$, there holds $\lim_{t\to\infty} x(t)=\vect 0_n$ and for all $i\in \mathcal{V}$, $\lim_{t\to\infty}g_i(t)=\bar g_i$ for some $\bar g_i\in[1,\infty)$.
\end{theorem}

The full proof is omitted as it follows similarly to the proof of Theorem~\ref{thm:full_infection}. Here, we briefly comment on the main differences. First, observe that \eqref{eq:cont_system_recovery} yields $g_i(t) = g_i(0)+\alpha_i\int_0^t (x_i(s))^pds$. Hence, $g_i(t)$ for each $i$ is monotonically increasing and thus either converges to a finite value or tends to infinity\footnote{In comparison, in Theorem~\ref{thm:full_infection}, $g_i(t)$ was monotonically decreasing and either converges to a positive value or tends to $0$ as $t\to\infty$.} as $t\to\infty$.  For the former possibility, it follows that $x_i(t) \in \mathcal{L}^p$, where $p$ is the integer defining the $\phi_i$ function of the adaptive controller. The latter possibility is excluded by a contradiction argument, essentially identical to that used in Theorem~\ref{thm:full_infection}. The key difference is that the system in \eqref{eq:prob1_I_sisdynamics_vect} has $\tilde D$ and $\tilde G(t)\tilde B$ replaced with $\tilde D\tilde G(t)$ and $\tilde B$, respectively.
We omit the other details and computations as these are, mutatis mutandis, the same as those appearing in the proof of Theorem~\ref{thm:full_infection}. 

In a similar vein, the complementary versions of Propositions~\ref{lem:gaininequality}, \ref{prop:positive_finite_t} and \ref{prop:final_reproduction} for recovery rate control can be obtained. (This would involve obvious adjustment to consider \eqref{eq:simp_cont_recovery} instead of \eqref{eq:simp_contsystotal}). We briefly comment on the differences, while the details are omitted for brevity. Proposition~\ref{prop:positive_finite_t} holds identically for the system \eqref{eq:simp_cont_recovery}, while in Proposition~\ref{prop:final_reproduction}, we redefine $\mathcal{R}_\infty = \rho((D\bar G)^{-1}B)$. For Proposition~\ref{lem:gaininequality}, we instead obtain a \textit{lower bound}:
\begin{equation}\label{eq:g_bound_recovery}
    \lim_{t\to\infty} g_i(t) \geq \sqrt{\frac{\alpha_i (x_i(0))^p}{d_i p}}. 
\end{equation}
To compute this lower bound, notice that $\lim_{t\to\infty} g_i(t) = \bar g_i < \infty$ due to Theorem~\ref{thm:full_recovery}. Moreover, it is obvious that $g_i(t) \leq \bar g_i$ for all $t\geq  0$. This implies that $\dot{x}_i(t) \geq -d_i \bar g_i x_i(t)$,
from which it follows that $x_i(t) \geq x_i(0) e^{-d_i\bar g_i t}$. Taking both sides to the power of $p$ and multiplying both sides by $\alpha_i$ yields $\alpha_i(x_i(t))^p \geq \alpha_i[x_i(0)]^p e^{-pd_i\bar g_i t}$.
Integrating both sides from $0$ to $\infty$ yields the inequality $\alpha_i \int_0^\infty (x_i(t))^p dt \geq \alpha_i [x_i(0)]^p \int_0^\infty e^{-pd_i\bar g_i t} dt $, or equivalently  $\bar g_i \geq \frac{\alpha_i [x_i(0)]^p}{p d_i \bar g_i}$. We obtain \eqref{eq:g_bound_recovery} by rearranging the final inequality.


\begin{remark}
One can view the infection and recovery parameter control problems as complementary problems. In the infection control problem, we decrease the gain $g_i(t)$ to reduce the interaction between susceptible individuals of population $i$ and infected individuals of population $j$ (for $j$ such that $b_{ij}>0$), and this decrease continues until the natural recovery rate of all nodes enable the disease to be eliminated from the network. A key interest is to ensure that $g_i(t)$ does not decrease to $0$ as this would represent total mobility restriction in population~$i$ and result in a loss of strong connectivity of the graph $\mathcal{G}(\bar GB)$; we verified this could not occur by proving $\bar g_i > 0$. In the recovery control problem, in contrast, the gain $g_i(t)$ on the recovery rate increases until it is strong enough to overcome the network infection dynamics. Here, we demonstrate that $\bar g_i < \infty$, for otherwise $\lim_{t\to\infty} g_i(t) = \infty$ would imply the controller is not feasible in a real-world implementation. \hfill $\triangle$
\end{remark}


\subsection{Simulations For Full Network Control}\label{ssec:sim_fullnetwork}


\subsubsection{Real-world Transportation Network}\label{sssec:realworld_sim_full}
We now demonstrate the effectiveness of our proposed decentralised adaptive-gain controller on a real-world large-scale network structure. Namely, we consider an $n=107$ node network, where each node is a province in Italy, and the links represent individual mobility and travel between provinces. The network is adapted from Ref.~\cite{parino2021modelling}, shown in Fig.~\ref{fig:italy_network} and the full adjacency matrix is found at \url{https://github.com/mengbin-ye/bivirus}. The original network $\bar{\mathcal{G}}$ is a complete directed graph, i.e., the adjacency $\bar B$ is a positive matrix but it is not symmetric. The largest and smallest entries of $\bar B$ differed by several orders of magnitude; differences in commuting patterns were such that some routes were heavily trafficked while other routes were virtually unused. We replaced by the value zero those entries of $\bar B$ below a threshold value $\kappa$ (equivalent to removing edges from $\bar{\mathcal{G}}$) in order to obtain an irreducible but not positive $B$. Thus, the resulting $\mathcal{G}$ is strongly connected but not complete. Finally, we normalized $B$ to satisfy $B\vect 1_n = 2\vect 1_n$. The precise method is found in the code from the URL provided above. We set $D = I_{107}$, which yields $\mathcal{R}_0 = 2$. We selected $10$ `seed' nodes uniformly at random to spread the disease, and selected their initial infection fractions $x_i(0)$ from a uniform distribution $[0.2, 0.7]$. The adaptive controllers are set as $\phi_i = \alpha_i x_i$ for all $i$, with each $\alpha_i$ selected from a uniform distribution $[0.01, 2]$. Evidently, $p =1$.

The simulation result is presented in Fig.~\ref{fig:n107_italy_fullinfection} and \ref{fig:n107_italy_fullinfection_Rt}. To maintain clarity, we show the average infection level $\frac{1}{n}\sum x_i(t)$ (thick black line), and the node infection fraction $x_i(t)$ and gain $g_i(t)$ for 5 of the seed nodes and 10 additional randomly selected nodes from the non-seed set. Evidently, the disease is eliminated from the entire network (black line reaches $0$), while none of the gains (dotted lines) reach $0$. This is consistent with Theorem~\ref{thm:full_infection}. The limiting reproduction number of the network is $\mathcal{R}_\infty = 0.897 < 1$, consistent with Proposition~\ref{prop:final_reproduction}. We highlight the decentralised nature of our proposed controller, and importantly, the fact that no information about the network (e.g. full or partial knowledge of $D$ or $B$) was required; the gains $\alpha_i$ were randomly selected. We note that while elimination of the disease is guaranteed, our extensive additional simulations have revealed that the rate of convergence to the healthy state $x= \vect 0_n$ can depend heavily on the value of $\alpha_i$, and separately, on $p$. This dependence is nontrivial. For instance, if every node has a large $\alpha_i$ value, say $\alpha_i > 2$ for all $i$ and $p= 1$, then convergence to $x = \vect 0_n$ occurs rapidly. However, the presence of just a single node with $\alpha_i < 0.1$ is sufficient to significantly slow down the convergence speed for all nodes, which we conjecture is due to the strongly connected nature of the network. Generally speaking, and for fixed values of $\alpha_i$ across different simulations, smaller values of $p$ lead to faster convergence rates, and a lower peak average infection, $\sup_{t} \frac{1}{n} \sum_i x_i(t)$. This is because $x_i(t) \in [0,1]$ and thus the gain adapts faster for smaller values of $p$. Indeed, Proposition~\ref{prop:final_reproduction} established exponential convergence for $p = 1$ but no conclusions on convergence rates are available for $p > 1$.


\begin{figure*} 
\centering
\subfloat[Italy Network]{\includegraphics[width= 0.2\linewidth]{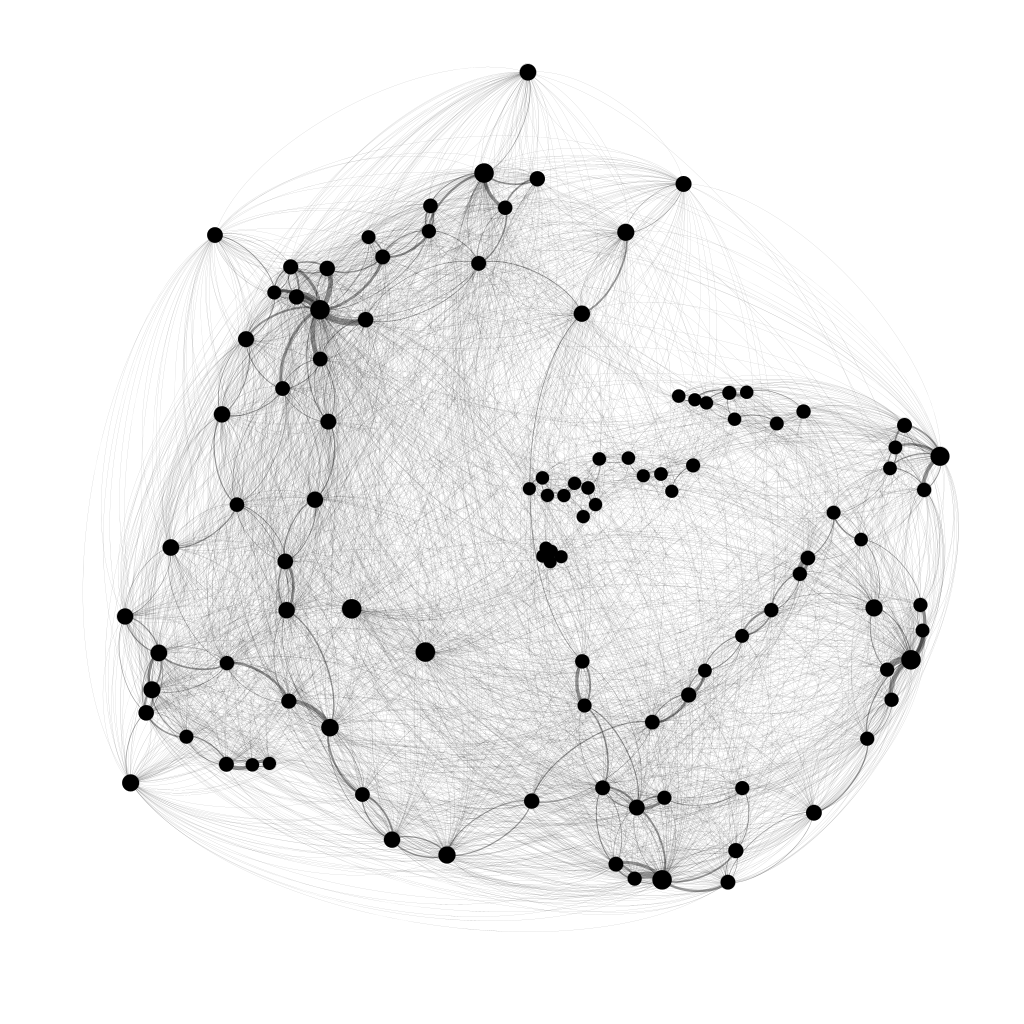}\label{fig:italy_network}}
\hfill
      \subfloat[Network dynamics]{\includegraphics[width= 0.4\linewidth]{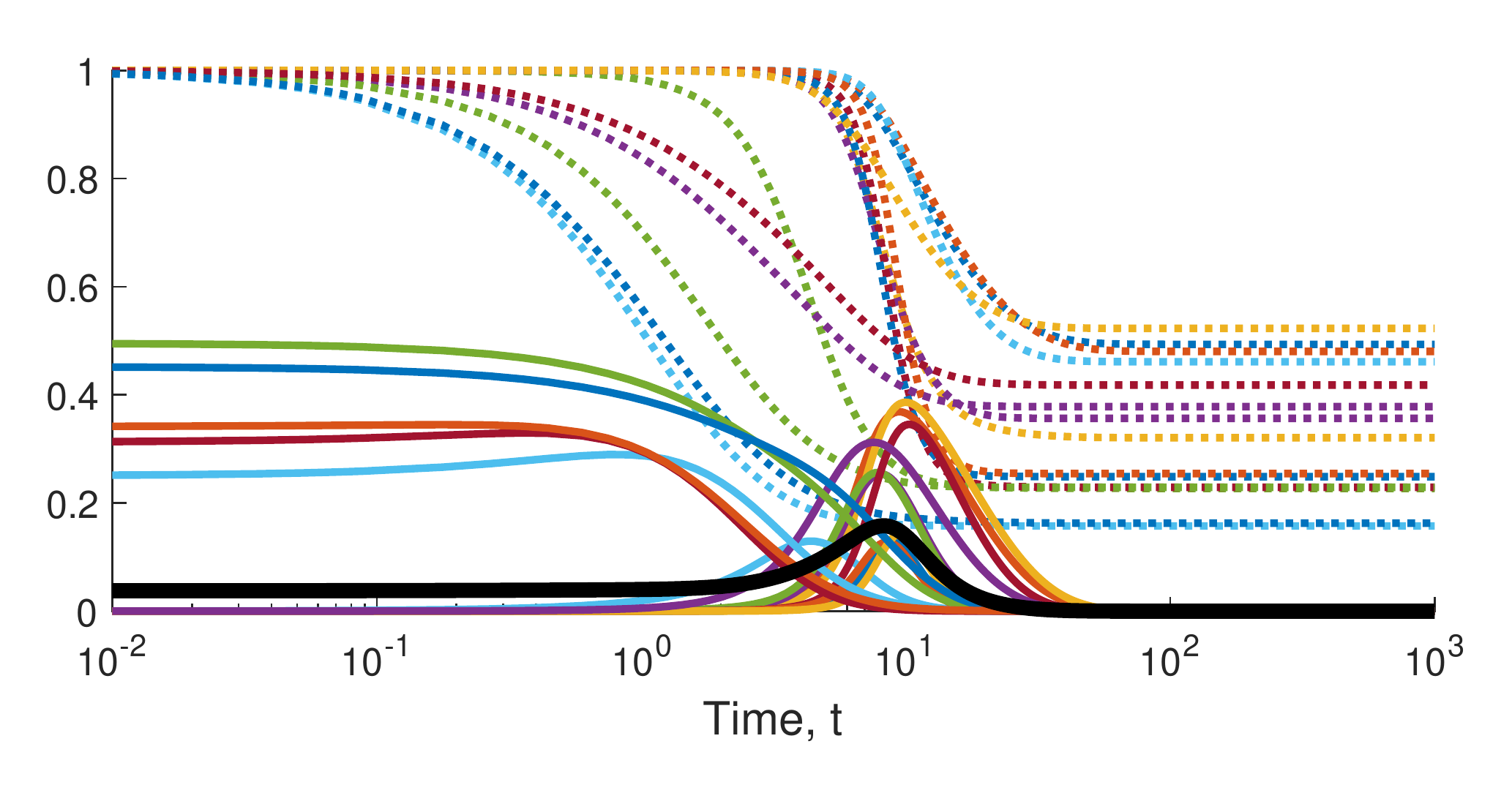}\label{fig:n107_italy_fullinfection}}
      \subfloat[Reproduction number $\mathcal{R}_t$]{\includegraphics[width= 0.4\linewidth]{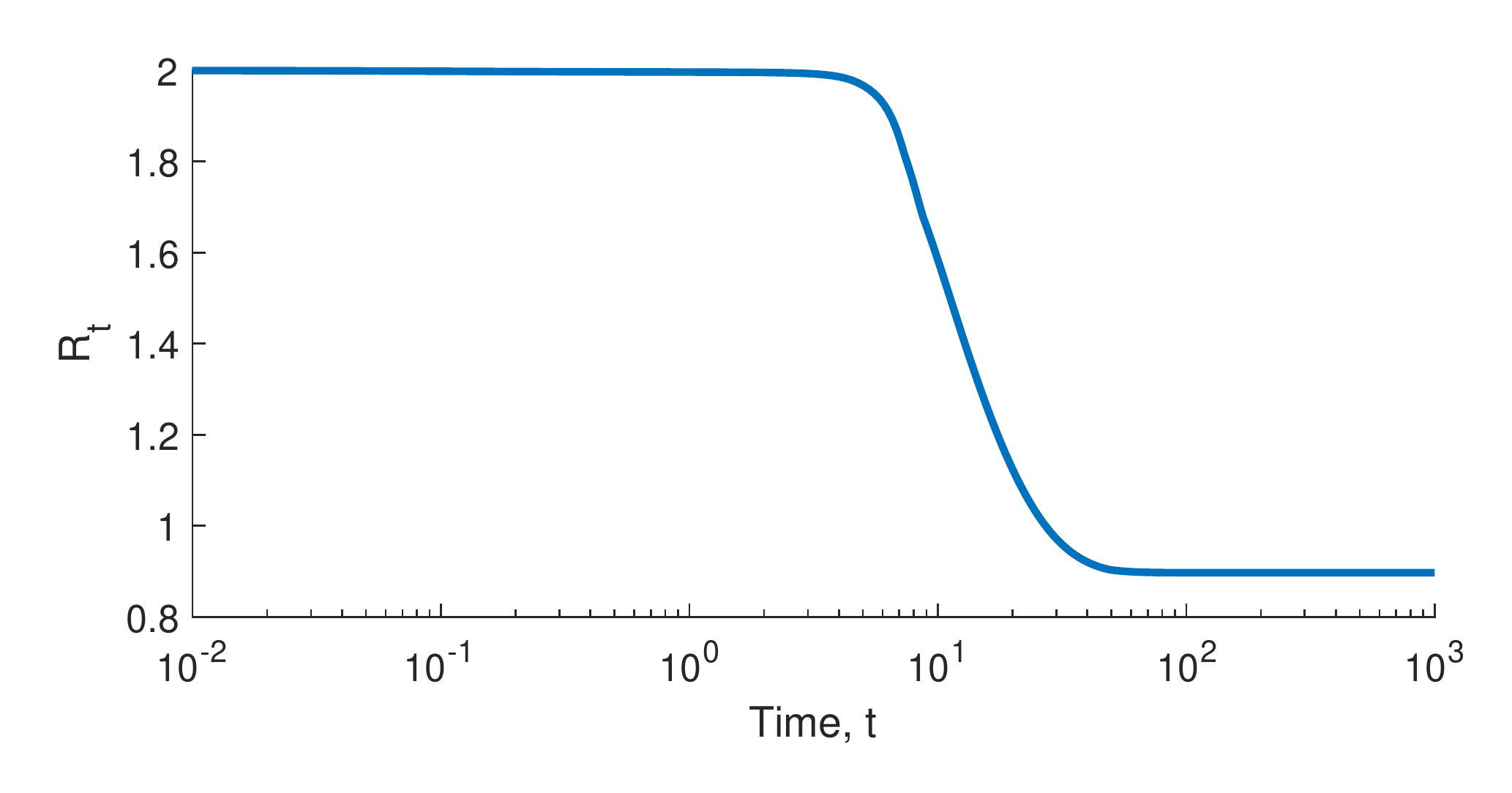}\label{fig:n107_italy_fullinfection_Rt}}
    \caption{Decentralised adaptive-gain infection rate control for the $n = 107$ mobility network of Italian provinces. In (a), the network structure is shown. In (b), the time evolution of average infection level $\frac{1}{n}\sum x_i(t)$ (thick black line), infection fraction $x_i(t)$ (solid coloured lines) and gain $g_i(t)$ (dotted coloured lines) are shown for a subset of the nodes in the network. In (c), the controlled reproduction number $\mathcal{R}_t$, as defined above Proposition~\ref{prop:final_reproduction}, is shown. Note the logarithmic scale of $t$, on the horizontal axis.  }    \label{fig:italy}
\end{figure*}

\subsubsection{Additional simulations}

This section contains additional simulations. For consistency, and to enable a comparison between different controllers, we use the same $x_i(0)$ and $\alpha_i$ as in Section~\ref{sssec:realworld_sim_full}.

First, we consider the same network scenario described in Section~\ref{sssec:realworld_sim_full} but with $p = 2$. The results are shown in Fig.~\ref{fig:italy_extra1}. Evidently, while the controller drives the infection at each node, $x_i(t)$ to $0$ and the limiting gains are positive, the limiting $\mathcal{R}_\infty$ appears to be approaching $1$ asymptotically, which differs from the case of $p =1$. Note the simulation time in Fig.~\ref{fig:italy_extra1} is several orders of magnitude greater than that in Fig.~\ref{fig:italy}. For $p = 1$ (Fig.~\ref{fig:italy}), the peak average infection level was $\sup_t \frac{1}{n}\sum_i x_i(t) = 0.158$, whereas for $p = 2$ (Fig.~\ref{fig:italy_extra1}), the peak average infection level was $0.300$.

We next consider the same network scenario described in Section~\ref{sssec:realworld_sim_full}, but using the adaptive recovery rate control. We keep $p = 1$, and simply switch from the controller in \eqref{eq:adaptive_law} to \eqref{eq:adaptive_law_recov}. The simulation outputs are shown in Fig.~\ref{fig:italy_extra2}. Here, we see that the disease is eliminated from every node in the network, while the gains converge to finite values. Similar to the case of Fig.~\ref{fig:italy}, we have $\mathcal{R}_\infty = 0.970 < 1$. The peak average infection level was~$0.180$.

Next, we consider the possibility of using piecewise constant and periodically updated adaptive gains. Although this scenario is not addressed theoretically, we provide here preliminary simulations to show that even under periodic updating, the proposed method shows promise and hence may be an interesting line of future work. In real-world implementation of public health interventions and measures to control epidemics, it is often the case that the policymakers roll out certain interventions/measures that are kept in place for weeks/months. Observations are made on how these interventions are impacting the epidemic spreading process, and after reevaluation, new interventions are implemented (either more severe if the epidemic is still spreading strongly or less severe if the epidemic is receding). Introducing interventions in phases such as this also allows for the population and medical staff to familiarise themselves with the interventions; constantly changing interventions may create significant logistical and implementation challenges. Toward this end, we may adjusting the dynamics in \eqref{eq:cont_system_infection}. In particular, with $k \in \mathbb N_0$ being a nonnegative integer, and $T > 0$ being the updating period, we propose that
\begin{subequations}\label{eq:cont_system_infection_periodic}
\begin{align}
\dot{x}_i(t) &= -d_ix_i(t) + \left(1-x_i(t)\right) g_i(kT) \sum_{j=1}^n{b_{ij}x_j(t)}, \\ & \qquad \qquad \forall\, t \in [kT, (k+1)T) \nonumber \\
\dot{g}_i(t) &= -\phi_i(x_i(t))g_i(t), \ \ g_i(0)=1,
\end{align}
\end{subequations}
In other words, we allow a background calculation for the gain to adapt continuously, but we update the effect of the gain on the epidemic dynamics $\dot{x}_i(t)$ periodically, with period length $T$. Thus, $g_i(t)$ appears in the $\dot{x}_i(t)$ as a piecewise constant gain. The scenario considered in Fig.~\ref{fig:italy} is simulated, but with the periodically updating dynamics as described in \eqref{eq:cont_system_infection_periodic}, with period $T = 5$ (Fig.~\ref{fig:n107_italy_fullinfection_period5} and $T = 10$ (Fig.~\ref{fig:n107_italy_fullinfection_period10}. Here, we see that even with periodic updating of the gains and for larger values of $T$, the disease is eliminated, although the transient dynamics can be significantly different. As $T$ becomes smaller and smaller, the difference in trajectories between the periodically updating system and continuously updating system reduces.

\begin{figure*} 
\centering
      \subfloat[Network dynamics]{\includegraphics[width= 0.4\linewidth]{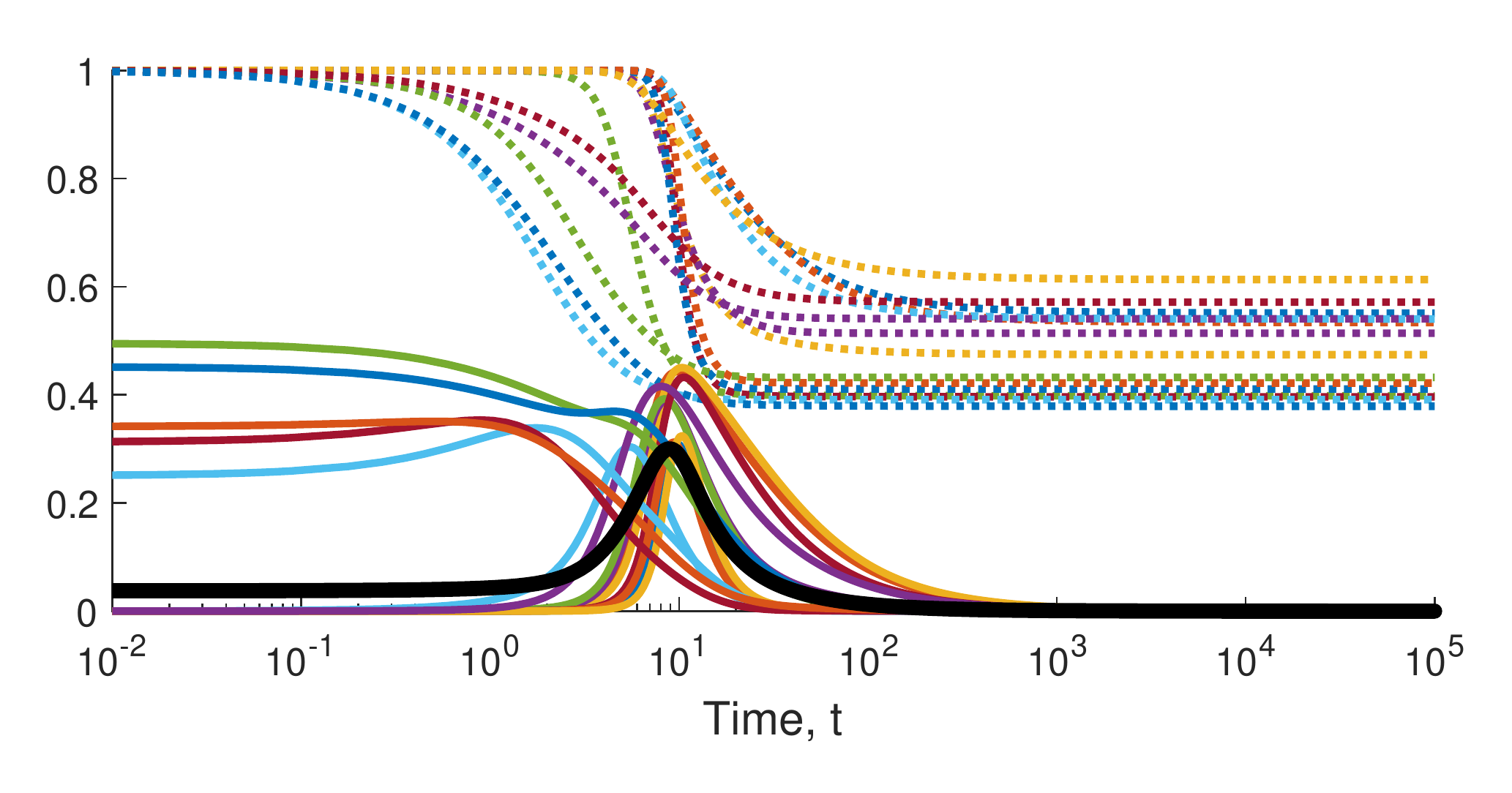}\label{fig:n107_italy_fullinfection_p2}}
      \subfloat[Reproduction number $\mathcal{R}_t$]{\includegraphics[width= 0.4\linewidth]{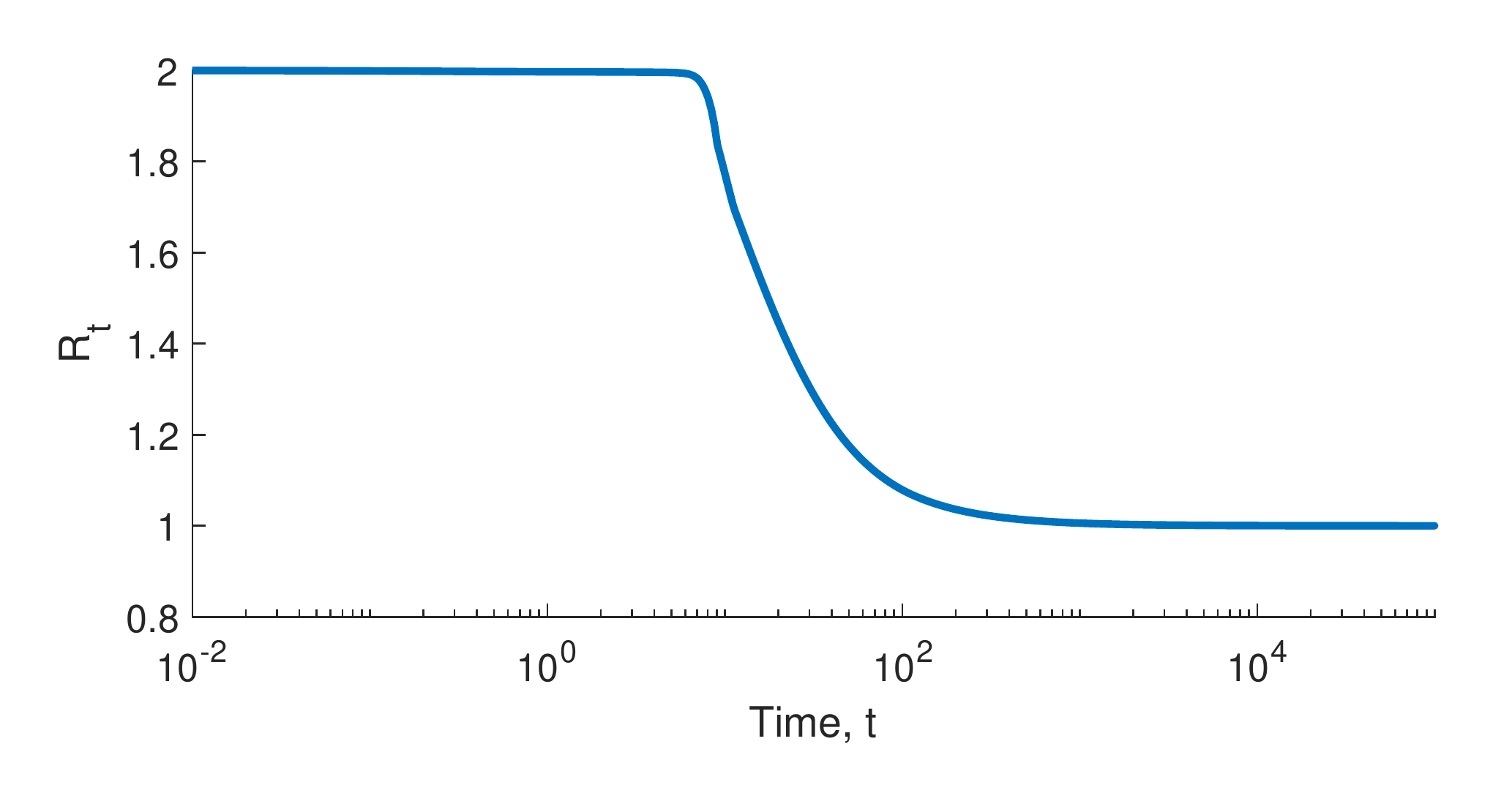}\label{fig:n107_italy_fullinfection_p2_Rt}}
    \caption{Decentralised adaptive-gain infection rate control for the $n = 107$ network of Italian provinces, with $p = 2$. In (a), the time evolution of average infection level $\frac{1}{n}\sum x_i(t)$ (thick black line), infection fraction $x_i(t)$ (solid coloured lines) and gain $g_i(t)$ (dotted coloured lines) are shown for a subset of the nodes in the network. In (b), the controlled reproduction number $\mathcal{R}_t$ is defined above Proposition~\ref{prop:final_reproduction}. }    \label{fig:italy_extra1}
\end{figure*}

\begin{figure*} 
\centering
      \subfloat[Network dynamics]{\includegraphics[width= 0.4\linewidth]{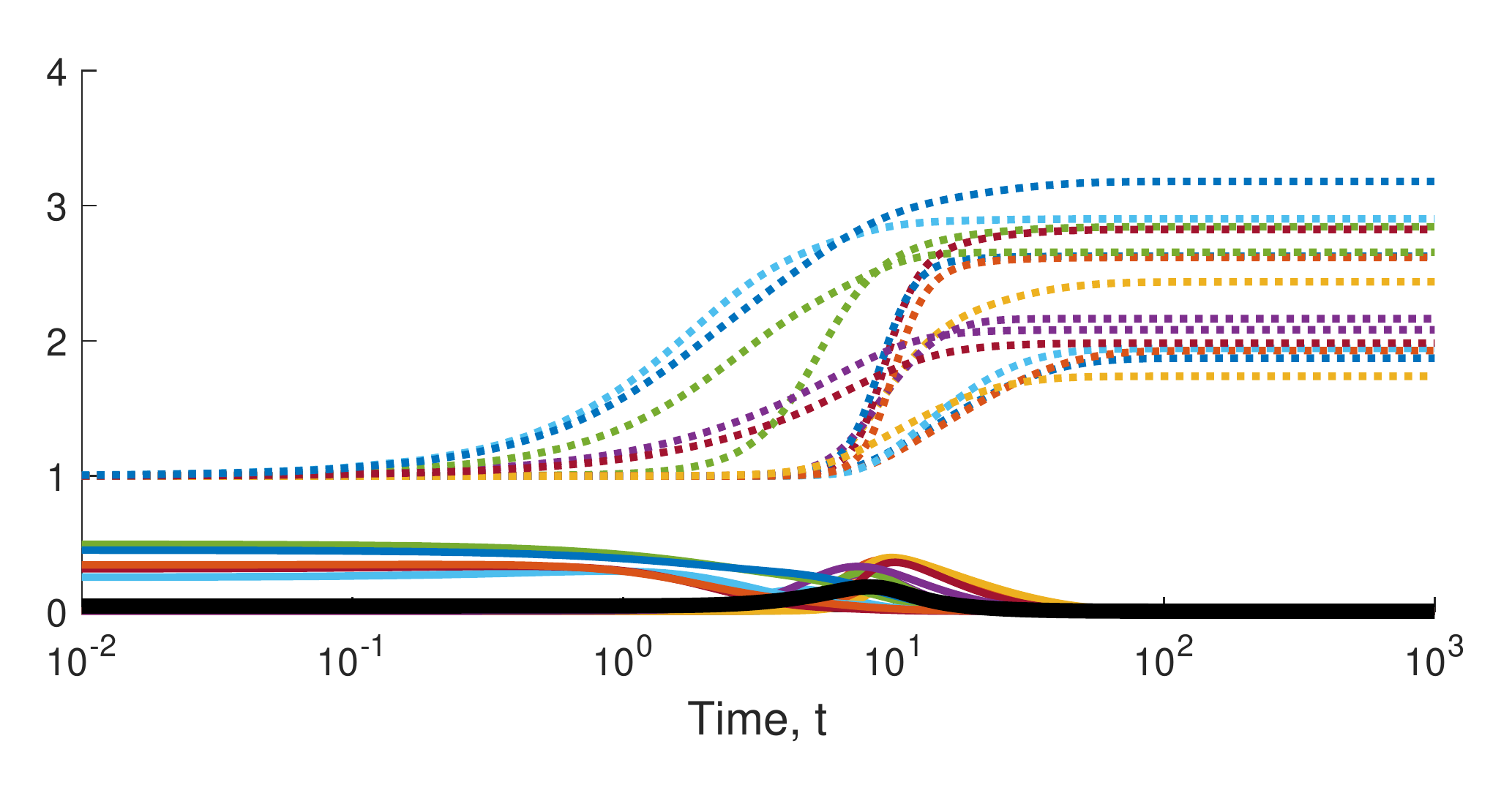}\label{fig:n107_italy_fullrecovery}}
      \subfloat[Reproduction number $\mathcal{R}_t$]{\includegraphics[width= 0.4\linewidth]{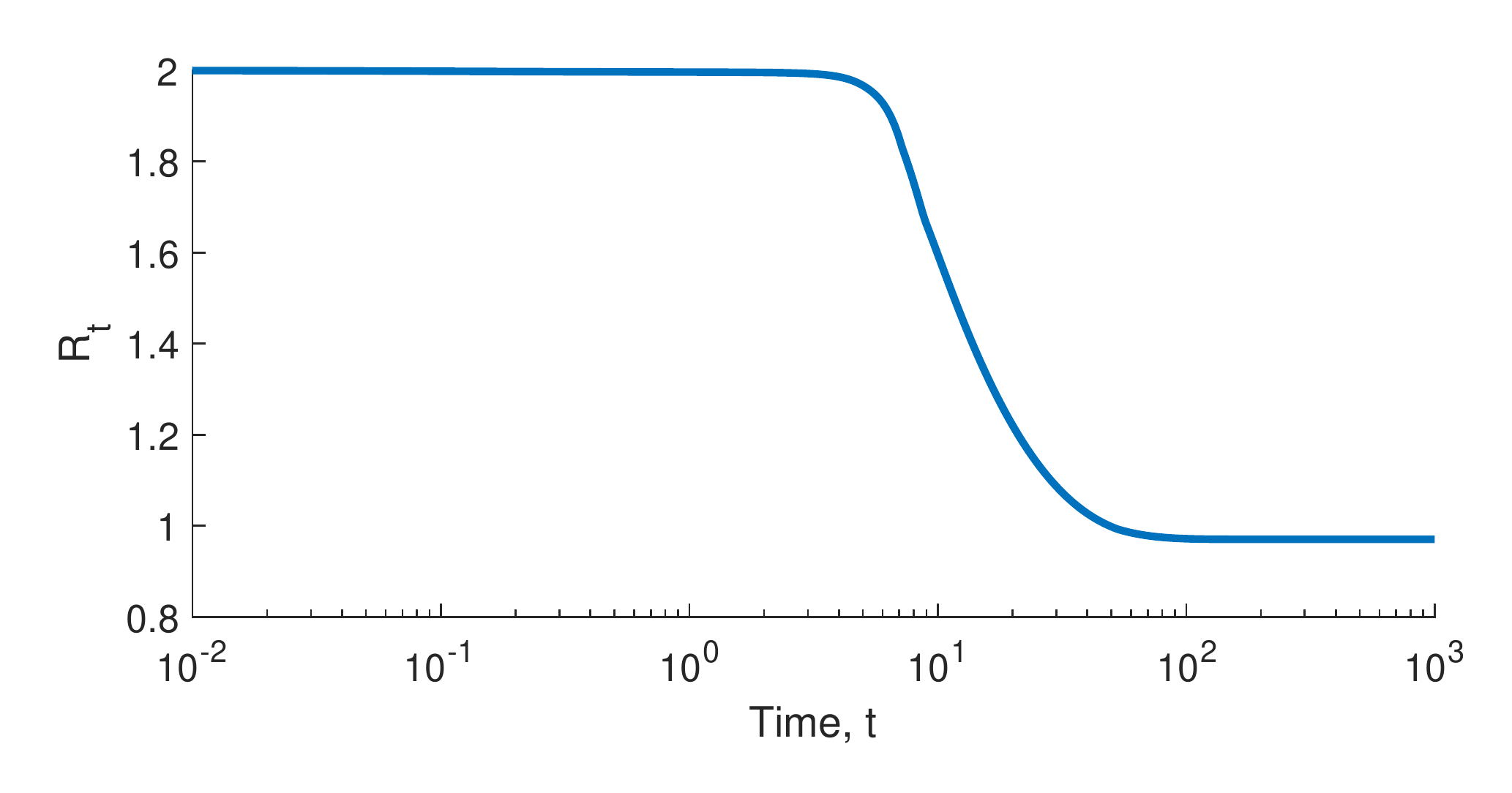}\label{fig:n107_italy_fullrecovery_Rt}}
    \caption{Decentralised adaptive-gain recovery rate control for the $n = 107$ network of Italian provinces. In (a), the time evolution of average infection level $\frac{1}{n}\sum x_i(t)$ (thick black line), infection fraction $x_i(t)$ (solid coloured lines) and gain $g_i(t)$ (dotted coloured lines) are shown for a subset of the nodes in the network. In (b), the controlled reproduction number $\mathcal{R}_t$ is shown. }    \label{fig:italy_extra2}
\end{figure*}

\begin{figure*} 
\centering
      \subfloat[Network dynamics, $T = 5$]{\includegraphics[width= 0.4\linewidth]{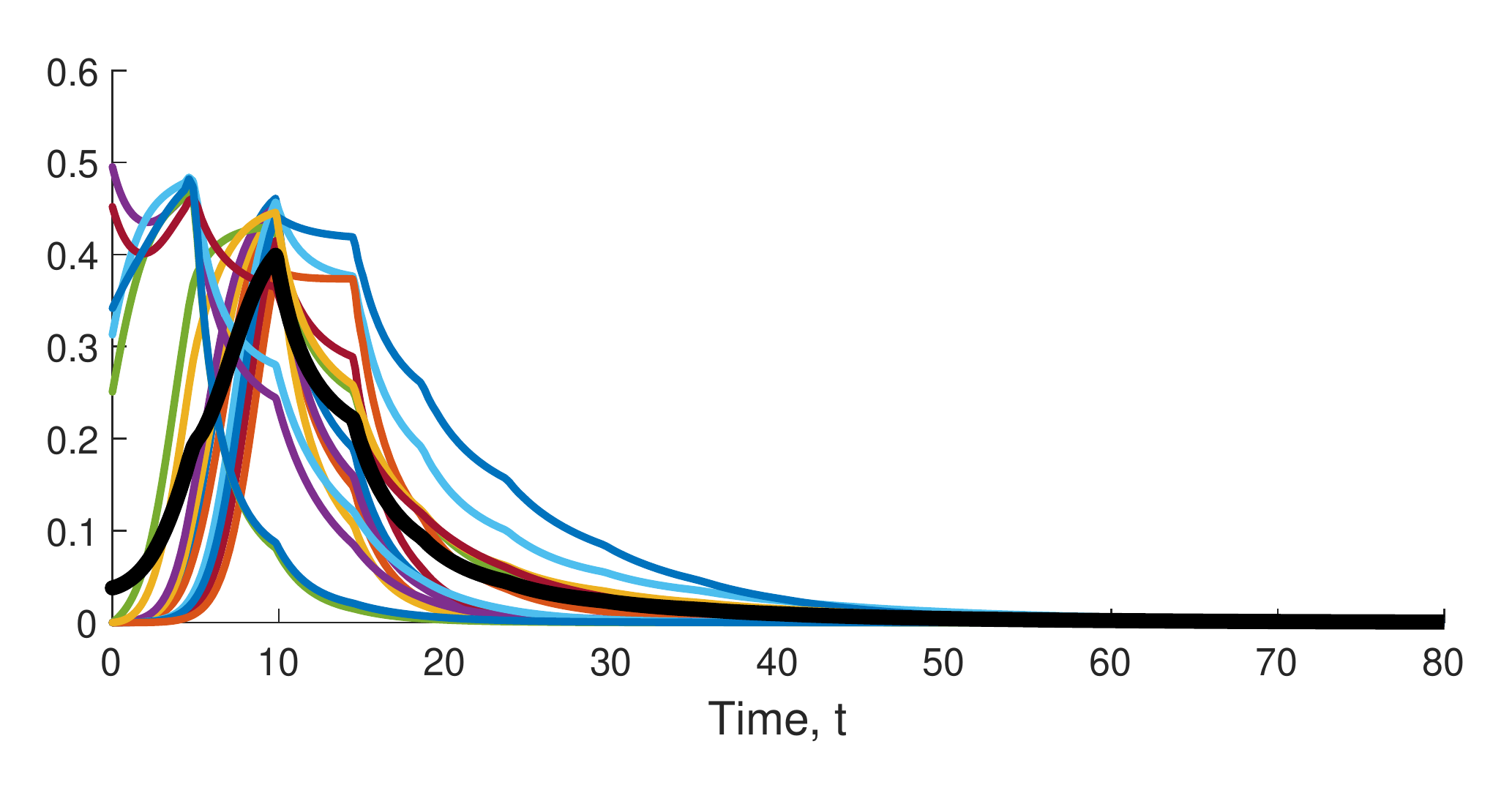}\label{fig:n107_italy_fullinfection_period5}}
      \subfloat[Network dynamics, $T = 10$]{\includegraphics[width= 0.4\linewidth]{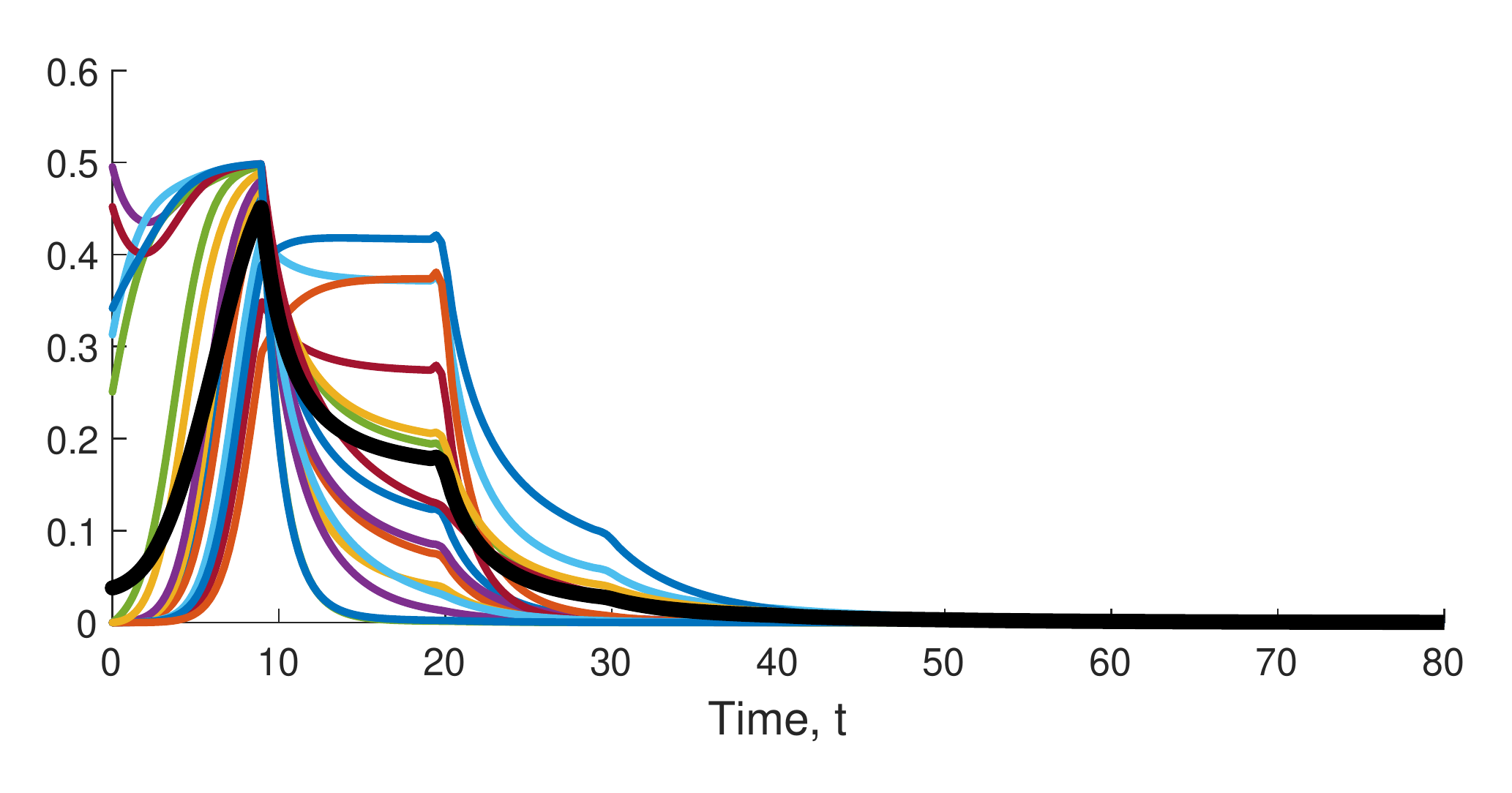}\label{fig:n107_italy_fullinfection_period10}}
    \caption{Decentralised adaptive-gain infection rate control for the $n = 107$ network of Italian provinces, with periodically updating gains. In (a), the period is set to $T = 5$, while in (b), the period is $T = 10$. For simplicity, we only show the time evolution of average infection level $\frac{1}{n}\sum x_i(t)$ (thick black line) and infection fraction $x_i(t)$ (solid coloured lines). The gains are omitted. }    \label{fig:italy_extra3}
\end{figure*}

\section{Partial Network Control}\label{sec:partial_control}

We now turn our attention to the problem of partial network control. In other words, we only control a strict subset of the nodes. A primary motivation for this section is to explore the circumstances under which one can still eliminate the disease from the network without requiring every node enact an adaptive-gain controller (which obviously reduces demand on policy and health efforts for combating the disease). We identify a necessary and sufficient condition for disease elimination when the controlled nodes are a strict subset of the total node set $\mathcal{V}$. We then propose an iterative algorithm that is guaranteed to identify a suitable set of nodes to control.

\subsection{Partial Infection Rate Control}\label{ssec:partial_infection}

The partial infection rate control problem builds on the problem explored in Section~\ref{ssec:full_infection}. Namely, we continue to consider the adaptive-gain controllers given in \eqref{eq:adaptive_law}, and implemented as in \eqref{eq:cont_system_infection}. However, we now relax Assumption~\ref{ass:phi_properties} as follows: 


\begin{assumption}[Properties of $\phi_i$]\label{ass:phi_properties_partialcont}
For some positive integer $p \in \mathbb N_+$, there holds $\phi_i(x_i) = \alpha_i {x_i}^{p}$ with tuning parameter $\alpha_i \geq 0$ for every $i\in \mathcal{V}$, and there exists at least one $j\in \mathcal{V}$ such that $\alpha_j = 0$.
\end{assumption}

With this assumption in place, we say that node~$i$ is controlled if $\alpha_i > 0$, and is uncontrolled if $\alpha_i = 0$. We define $\mathcal{C}$ to be the set of controlled nodes, i.e., 
\begin{equation}
\mathcal{C} \triangleq \{i\in\mathcal{V}\, |\, \alpha_i > 0\},
\end{equation}
and the set of uncontrolled nodes as $\mathcal{U} \triangleq \mathcal{V}\setminus \mathcal{C}$. With such a definition, we can effectively study the same dynamical system in \eqref{eq:simp_contsystotal}, but there are now significant new challenges arising due to the uncontrolled nodes. 

The partial infection rate control problem can be summarised as follows.

\begin{problem}\label{prob:partialinfcontrol}
Consider the system in \eqref{eq:simp_contsystotal} under Assumptions~\ref{assm:strongly_connected}, \ref{assm:endemic_R0} and \ref{ass:phi_properties_partialcont}. 
\begin{enumerate}
    \item Demonstrate the existence of, and identify conditions on $D$ and $B$ such that, if the conditions are satisfied, then there exists a proper subset $\mathcal{C} \subset \mathcal{V}$, with decentralised controllers given in \eqref{eq:adaptive_law} for each $i\in \mathcal{C}$, that yields $\lim_{t\to\infty}{x(t)}=\mathbf{0}_n$ for any $x(0)\in \Xi_n$, and $\lim_{t\to\infty} g(t) > \vect 0_n$.
    \item If such a set $\mathcal{C}$ exists, then develop an iterative algorithm that identifies a set $\tilde{\mathcal{C}}$ of nodes to control, that yields $\lim_{t\to\infty}{x(t)}=\mathbf{0}_n$ for any $x(0)\in \Xi_n$, and $\lim_{t\to\infty} g(t) > \vect 0_n$.
\end{enumerate}
\end{problem}



The problem has two parts which we will separately address. The first part is concerned with determining the following: given an SIS network with prescribed pair $(D, B)$, does there exists a set of nodes $\mathcal{C}$ that if controlled will result in elimination of the disease, with $\mathcal{C}$ a strict subset of $\mathcal{V}$. As we will show in the sequel, existence of $\mathcal{C}$ depends on $D$ and $B$, because there are in fact some pairs of $(D, B)$ that require every node to be controlled to eliminate the disease, and then partial network control is never achievable. The second part of the problem is a design problem: development of a computationally tractable method for identifying a suitable set $\mathcal{C}$ when such a set exists.

\subsubsection{Conditions for Existence of Control Node Set $\mathcal{C}$}\label{sssec:controllable_infection}

In order to address the first part of the problem, we first present the following supporting result.

\begin{proposition}\label{prop:db_partition}
Let $B$ and $D$ be defined as in \eqref{eq:sisdynamics_vect}, and partitioned as
    \begin{equation}\label{eq:DB_partition}
        B=\begin{bmatrix} B_{11}&B_{12}\\B_{21}&B_{22}
\end{bmatrix}\quad D=\begin{bmatrix}D_1&\mat 0_{k\times (n-k)} \\\mat 0_{(n-k)\times k}&D_2\end{bmatrix},
    \end{equation}  
    with $B_{11}$ and $D_1$ being $k\times k$ in size, and $B_{22}$ and $D_2$ being $(n-k)\times (n-k)$ in size. Then, there exists a positive diagonal $\bar G_2 \in\mathbb R^{(n-k)\times(n-k)}$ such that 
    \begin{equation}\label{eq:DbarGB_matrix}
     -D + \bar G B = -\begin{bmatrix}D_1&\mat 0 \\\mat 0&D_2\end{bmatrix} + \begin{bmatrix} I_{k} & \mat 0 \\\mat 0&\bar G_2
\end{bmatrix}\begin{bmatrix} B_{11}&B_{12}\\B_{21}&B_{22}
\end{bmatrix}
    \end{equation}
    is Hurwitz if and only if $D_1-B_{11}$ is a nonsingular $M$-matrix, where $\bar G$ has obvious definition.
\end{proposition}
\begin{proof}
    Suppose firstly that $\bar G_2$ exists such that $-D+\bar GB$ is Hurwitz. Then $D-\bar GB$ is a nonsingular $M$-matrix.
    Then the principal submatrix $D_1-B_{11}$ is a nonsingular $M$-matrix, see~\cite[p.~156]{berman1979nonnegative_matrices}, and has eigenvalues with strictly positive real part. Hence, $-D_1 + B_{11}$ is Hurwitz.

For the converse, suppose that $D_1-B_{11}$ is a nonsingular $M$-matrix, implying $-D_1+B_{11}$ is Hurwitz, and let $\tilde G_2 \in \mathbb R^{(n-k)\times(n-k)}$ be an arbitrary diagonal matrix with positive diagonal elements. For $\epsilon\in[0,1]$, define
\[
\tilde G(\epsilon)=\begin{bmatrix} I_k&\mat 0\\\mat 0&\epsilon \tilde G_2\end{bmatrix}.
\]
As $\epsilon\downarrow 0$, the matrix 
\[
-D+\tilde G(\epsilon)B=\begin{bmatrix}
-D_1+B_{11}&B_{12}\\
\epsilon\tilde G_2 B_{21}&-D_2+\epsilon\tilde G_2B_{22}
\end{bmatrix},
\]
approaches a matrix whose eigenvalues are those of $-D_1+B_{11}$ and $-D_2$, which are both Hurwitz matrices. Choose $\epsilon$ so that $-D+\tilde G(\epsilon) B$ is Hurwitz, and set $\bar G=\tilde G(\epsilon)$. Then $-D+\bar G B$ is Hurwitz and has the desired structure. 
\end{proof}




In thinking about the partial network control problem, given a set of control nodes $\mathcal{C}$, we can without loss of generality reorder the nodes such that $\mathcal{U} = \{1, \hdots, k\}$ and $\mathcal{C} = \{k+1, \hdots, n\}$, for some integer $k \geq 1$, and the partitioning in \eqref{eq:DB_partition} identifies the uncontrolled and controlled subnetworks. We now state the main result of this section, which links the dynamics of the overall SIS network with the stability of the uncontrolled network, governed by the stability of the matrix $-D_1+B_{11}$. Subsequently, we provide a necessary and sufficient condition for the existence of a nonempty set $\mathcal{U}$.

\begin{theorem}\label{thm:partial_infection}
    Consider the system in \eqref{eq:simp_contsystotal} under Assumptions~\ref{assm:strongly_connected}, \ref{assm:endemic_R0} and \ref{ass:phi_properties_partialcont}. Without loss of generality, let the nodes be ordered as $\mathcal{U} = \{1, \hdots, k\}$ and $\mathcal{C} = \{k+1, \hdots, n\}$, with $B$ and $D$ partitioned as in \eqref{eq:DB_partition}. Then the following statements are equivalent.
    \begin{enumerate}
        \item For all $\xi(0)\in \Xi_n \times \vect 1_n$, there holds $\lim_{t\to\infty} x(t) = \vect 0_n$ and $\lim_{t\to\infty} g(t) = \bar g$, where $\bar g > \vect 0_n$.
        \item The matrix $- D_1 + B_{11}$ is Hurwitz.
    \end{enumerate}
\end{theorem}

\begin{proof}
    We first prove that Item~1) implies Item~2), by contradiction. First, notice that Item~1) implies that the trajectory $\xi(t)$ of \eqref{eq:simp_contsystotal} approaches $\bar \xi\triangleq \xi(\infty) =[\vect 0_n^\top, \bar g^\top]^\top$ asymptotically. The differential equation is time-invariant, and there is a well-defined Jacobian matrix $J(\bar \xi)$ at the equilibrium point $\bar \xi$. Using \eqref{eq:f_infection}, we can compute 
\begin{equation}
J(\bar \xi)=\begin{bmatrix}
-D+\bar GB&{\bf{0}}_{n\times n}\\
J_{21}&{\bf{0}}_{n\times n}
\end{bmatrix},
\end{equation}
where $J_{21}={\rm{diag}}(\alpha_i\bar g_i)$ if $p=1$ and $J_{21}={\bf 0}_{n\times n}$ if $p\geq2$. The matrix $\bar G$ has the obvious definition from \eqref{eq:DbarGB_matrix}.
The eigenvalues of the Jacobian are the eigenvalues of the Metzler matrix $Q\triangleq -D+\bar GB$ together with $n$ occurrences of the zero eigenvalue. The argument proving Item~1 is the same, mutatis mutandis, as the argument for proving Proposition~\ref{prop:final_reproduction}, and is omitted. 

Next, observe that if $-D_1+B_{11}$ is not Hurwitz, then either i) $s(-D_1+B_{11}) > 0$ (call this Case 1) or ii) $s(-D_1+B_{11}) = 0$ (call this Case 2). Since $-D_1 + B_{11}$ is a Metzler matrix, if $s(-D_1+B_{11}) = 0$, then there can be no other eigenvalues on the imaginary axis other than those at the origin. We address these cases separately.

\textit{Case 1:} Now, assume to obtain a contradiction that 
$s(-D_1+B_{11}) > 0$. According to Proposition~\ref{prop:db_partition}, $s(Q)> 0$ for any positive diagonal $\bar G$. However, this contradicts our conclusion above, which established that $s(Q) \leq 0$.

\textit{Case 2:} Now, assume to obtain a contradiction that $s(-D_1+B_{11}) = 0$. Recall that $B$ is irreducible (Assumption~\ref{assm:strongly_connected}). Item~1 implies that $\bar G$ is a positive diagonal matrix, which further implies that $Q = -D + \bar GB$ is an \textit{irreducible} Metzler matrix. The fact that $Q$ is an irreducible Metzler matrix, and $s(Q) \leq 0$ as established above, implies that $-Q$ is an irreducible $M$-matrix. According to \cite[Theorem~5.7]{fiedler1962matrices}, if $-Q$ is an irreducible $M$-matrix, then all proper principal minors of $-Q$ are positive. However, the determinant of $D_1-B_{11}$ is one such principal minor and it is zero by assumption, which creates the contradiction.

We now prove that Item~2) implies Item~1). First, and similarly to the proof of Theorem~\ref{thm:full_infection}, we can partition the node set $\mathcal{C}$ into two disjoint sets of $\mathcal{C}_{I} \triangleq \{k+1, k+2, \hdots, r\}$ and $\mathcal{C}_F \triangleq \{r+1, r+2, \hdots, n\}$, with the property that $\int_0^\infty{\phi_i(x_i(s))}ds$ is infinite for all $i \in \mathcal{C}_I$ and finite for all  $i\in\mathcal{C}_F$. For the moment both extreme cases of $r=k$ and $r=n$ are allowed. Note that $\lim_{t\to\infty} g_i(t) = \bar g_i > 0$ for all $ i \in \mathcal{C}_F$, while $\lim_{t\to\infty} g_i(t) = 0$ for all $ i \in \mathcal{C}_I$.

Identically to the proof of Theorem~\ref{thm:full_infection}, we can conclude that for all $i \in \mathcal{C}_F$, there holds $x_i(t) \in \mathcal{L}^p$ and further $\lim_{t\to\infty} x_i(t) = 0$. To complete the proof, we will first prove that for all $i \in \mathcal{U} \cup \mathcal{C}_I$ there holds $\lim_{t\to\infty} x_i(t) = 0$ and $x_i(t) \in \mathcal{L}^p$, from which we will be able to show that $\lim_{t\to\infty} g_i(t) > 0$ for $i\in \mathcal{C}_I$. This final property creates a contradiction, and hence $\mathcal{C}_I$ is in fact empty. 

Let us define $\bar x = [x_1, \hdots, x_k]^\top$ and $\tilde x = [x_{k+1}, \hdots, x_{r}]^\top$ and $\hat x = [x_{r+1}, \hdots, x_n]^\top$. As mentioned above, we wish to study $x_i(t)$ for $i \in \mathcal{U} \cup \mathcal{C}_I$; the dynamics are given by
\begin{align}\label{eq:subsystem_dynamics}
    &\begin{bmatrix}
        \dot{\bar x}(t) \\ \dot{\tilde x}(t)
    \end{bmatrix}  = \Bigg(\!\!-\begin{bmatrix}
        D_{11} & \\ & \tilde D 
    \end{bmatrix}+\begin{bmatrix}
        (I -\bar X(t)) & \\ &  \!(I-\tilde X(t))\tilde G(t) 
    \end{bmatrix} \nonumber \\
     & \quad \times\!\begin{bmatrix}
        B_{11}\! \!& \!\tilde B_{12} \\  \tilde B_{21}\!\! &\!  \tilde B_{22}
    \end{bmatrix}\!\Bigg)\!
    \begin{bmatrix}
        \hat x(t) \\ \tilde x(t)
    \end{bmatrix}\!+\!\begin{bmatrix}
        (I -\bar X(t))\hat B_1 \hat x(t) \\ (I\!-\!\tilde X(t))\tilde G(t) \hat B_2 \hat x(t)
    \end{bmatrix}
\end{align}
Here, $\bar X = \diag(x_1, \hdots, x_k)$ and $\tilde X = \diag(x_{k+1}, \hdots x_r)$, while $\tilde D = \diag(d_{k+1}, \hdots, d_{r})$ and $\tilde G = \diag(g_{k+1}, \hdots, g_r)$. The matrices $\tilde B_{12}$, $\tilde B_{21}$, and $\tilde B_{22}$ are block submatrices of $B$ that capture the edges from nodes in $\mathcal{C}_I$ to nodes in $\mathcal{U}$, from nodes in $\mathcal{U}$ to nodes in $\mathcal{C}_I$, and from nodes in $\mathcal{C}_I$ to each other, respectively. Similarly, $\hat B_1$ and $\hat B_2$ are block submatrices of $B$ that capture edges from nodes in $\mathcal{C}_F$ to nodes in $\mathcal{U}$ and $\mathcal{C}_F$, respectively. Note that we have omitted the dimensions of the two $I$ matrices for brevity, these being obvious from the context.

Recall that $\lim_{t\to\infty}g_i(t) = 0$ for all $i\in\mathcal C_I$ by definition. Thus, it follows that for any $\epsilon > 0$, there exists some $\tau_i\geq 0$ such that $0\leq g_i(t)\leq \epsilon$ whenever $t\geq \tau_i$. Then for $t\geq \tau:=\max_{i=1, \dots, k}{\tau_i}$, it holds that $\Tilde{G}(t)\leq \epsilon I$. For $t\geq\tau$, we obtain from \eqref{eq:subsystem_dynamics} the following inequality:
\begin{align*}\label{eq:subsystem_ineq}
    \begin{bmatrix}\dot{\bar x}(t) \\ \dot{\tilde x}(t)
    \end{bmatrix} \leq \begin{bmatrix}
        -D_{11} + B_{11} & \tilde B_{12} \\ \epsilon \tilde B_{21} & -\tilde D +\epsilon\tilde B_{22}
    \end{bmatrix}
    \begin{bmatrix}
        \hat x(t) \\ \tilde x(t)
    \end{bmatrix}+w(t),
\end{align*}
where $w(t) = [(\hat B_1 \hat x(t))^\top,  (\epsilon \hat B_2 \hat x(t))^\top]^\top$ is an input signal.

Define the Metzler matrix 
\begin{equation}
    A_\epsilon =\begin{bmatrix}
        -D_{11} + B_{11} & \tilde B_{12} \\ \epsilon \tilde B_{21} & -\tilde D +\epsilon\tilde B_{22}
    \end{bmatrix}.
\end{equation}
By hypothesis $-D_{11}+B_{11}$ is Hurwitz, and hence according to \cite[Corollary~1]{souza2017note}, $A_{\epsilon}$ is Hurwitz if and only if the matrix $Z_\epsilon = -\tilde D + \epsilon \tilde B_{22} - \epsilon\tilde B_{21}(-D_{11}+B_{11})^{-1}\tilde B_{12}$ is Hurwitz. Since $-D_{11}+B_{11}$ is Hurwitz and Metzler, it follows that $D_{11}-B_{11}$ is a nonsingular $M$-matrix, and thus its inverse is a strictly positive matrix~\cite{berman1979nonnegative_matrices}. In other words, $(-D_{11}+B_{11})^{-1}$ has all negative entries. From the fact that $\tilde B_{12}$, $\tilde B_{21}$ and $\tilde B_{22}$ are all nonnegative matrices, we can write $Z_{\epsilon} = -\tilde D + \epsilon C$ for some nonnegative matrix $C$. Since $\tilde D$ is diagonal with all positive entries, it is obvious that $Z_\epsilon$ is Hurwitz if $\epsilon$ is sufficiently small. Assume henceforth that such a choice of $\epsilon$ has been taken.

Consider the system $\dot{\Tilde{y}}(t) = A_\epsilon \tilde y(t) + w(t)$,
with $\tilde y(0)$ selected such that $\tilde y(\tau) = [\bar x(\tau)^\top, \tilde x(\tau)^\top]^\top$. Following essentially an identical argument to that used below \eqref{eq:positive_linear_system}, and hence omitted to avoid repetition, we can show that $\tilde y(t) \geq [\bar x(t)^\top, \tilde x(t)^\top]^\top$ for all $t\geq \tau$. From here, and adopting similar arguments to those used below \eqref{eq:positive_linear_system}, we conclude that $\bar x(t) \in \mathcal{L}^p_k$ and $\tilde x(t) \in \mathcal{L}^p_{r-k+1}$
However, this implies that for every $i \in\mathcal C_I$, $\int_0^t{\phi(x_i(s))} ds = \alpha_i \int_0^t x_i^p ds$ converges to a finite value as $t\to\infty$, and because $g_i(t) = g_i(0)e^{-\int_0^t \phi_i(x_i(s))ds}$, there exists $\bar g_i>0$ such that $g_i(t)\to\bar g_i$ as $t\to\infty$: a contradiction. Thus, we must have $r=k$, and $\mathcal{C}_I$ is empty. It follows that $\lim_{t\to\infty} x(t) = \vect 0_n$ and $\lim_{t\to\infty} g(t) = \bar g > \vect 0_{n}$ as claimed.
\end{proof}

We conclude by providing an auxiliary result, from which we can then derive a simple, node-based necessary and sufficient condition for the existence of a nonempty $\mathcal{U}$. 

\begin{proposition}\label{prop:necessary_partial}
    There exists a suitable proper subset $\mathcal{C} $ of $\mathcal{V}$ that solves the partial infection rate control problem (i.e., ensures that $-D_1+B_{11}$ is Hurwitz in Theorem~\ref{thm:partial_infection}) if and only if there exists $i\in\mathcal{V}$ such that $d_i > b_{ii}$. Moreover, if $j\in\mathcal V$ satisfies $d_j \leq b_{jj}$, then any suitable $\mathcal{C}$ must be such that $j\in\mathcal C$.
\end{proposition}
\begin{proof}
    To begin, we prove the first claim of the proposition. For sufficiency, assume there is a single node $i$ such that $d_i > b_{ii}$. Then evidently, $\mathcal{C} = \mathcal{V}\setminus \{i\}$ will ensure that $-D_1 + B_{11}$ as defined in Theorem~\ref{thm:partial_infection} is Hurwitz; elimination of the disease from the network with strictly positive limiting gains $\bar g_i$ is assured. For necessity, assume to obtain a contradiction that $d_i \leq b_{ii}$ for all $i\in\mathcal{V}$ and we have selected $\mathcal C$ such that $-D_1 + B_{11}$ is Hurwitz. From the theory of $M$-matrices, $D_1 - B_{11}$ is a nonsingular $M$-matrix (equivalently, $-D_1 + B_{11}$ is Hurwitz) if and only if every principal submatrix is a nonsingular $M$-matrix~\cite{berman1979nonnegative_matrices}. Yet, $d_1 - b_{11} \leq 0$ is a principal submatrix, and it is not a nonsingular $M$-matrix. This establishes the contradiction. The second (and final) claim of the proposition follows a similar proof to the proof of necessity and is omitted. 
\end{proof}


The above proposition provides a simple and intuitive necessary and sufficient condition for Problem~\ref{prob:partialinfcontrol} to be solvable. Namely, we require the existence of a node $i\in\mathcal V$ such that $d_i > b_{ii}$. If no such node exists, then one cannot find any proper subset $\mathcal{C}$ of controlled nodes that ensures $-D_1+B_{11}$ is Hurwitz, and which according to Theorem~\ref{thm:partial_infection} is equivalent to driving the network to the disease free state, $x = \vect 0_n$ while ensuring that $\bar g > \vect 0_n$. On the other hand, any node $j$ satisfying $d_j \leq b_{jj}$ must be in the controlled set of nodes.

Such a condition is intuitive, as we now elaborate. We can define the local reproduction number of a population (node) $i$ as $\mathcal{R}^i_0 = b_{ii}/d_i$. Then, $\mathcal{R}^i_0 < 1$ is the necessary and sufficient condition for $x_i(t) \to 0$ exponentially fast if node $i$ is isolated (has no incoming edges). We interpret our result as saying: we can omit controlling some nodes in the network only if at least one node has $\mathcal{R}^i_0 < 1$, i.e., at least one node can become disease free without control and without infections arriving from other nodes in the network. Meanwhile, we must control all nodes that have $\mathcal{R}^i_0 > 1$, i.e., that cannot eliminate the disease by themselves while isolated. Note that depending on the precise network structure and $D$ and $B$ parameter matrices, it may still be necessary to control some nodes which have $\mathcal{R}^i_0 < 1$; we provide such an example in our simulation in Section~\ref{ssec:sim_partialnetwork}.


\subsubsection{Algorithm for Identifying Node Set $\mathcal{C}$}\label{sssec:algorithm_control}

We now propose an iterative algorithm for identifying a suitable node set $\mathcal{C}$ for control. To begin, we introduce some additional notation and definitions pertaining to graphs, as well as a result on Metzler matrices which we will exploit.

A graph $\mathcal{G} = (\mathcal{V}, \mathcal{E}, A)$ with associated adjacency matrix $A$, will sometimes be expressed as $\mathcal{G}(A)$ for convenience. In this section, we also consider signed graphs, where the weight of an edge $(i,j) \in \mathcal{E}$ can be negative, and this is reflected in the associated entry of the adjacency matrix being  negative, $a_{ji} < 0$. Thus, $A$ does not have to be a nonnegative matrix. Given this graph, let $\mathcal{S} \subset \mathcal{V}$ be a proper subset of its nodes. The subgraph of $\mathcal{G}$ \textit{induced} by $\mathcal{S}$ is denoted as $\mathcal{G}[\mathcal{S}]:=(\mathcal{S}, \mathcal{E}_\mathcal{S}, A_\mathcal{S})$ with vertex set $\mathcal{S}$, edge set $\mathcal{E}_\mathcal{S}:=\{(j,i)\in\mathcal{S}\times\mathcal{S}: (j,i)\in\mathcal{E}\}$ and weighted adjacency matrix $A_\mathcal{S}\in\mathbb{R}^{|\mathcal{S}|\times |\mathcal{S}|}$ defined by $A_\mathcal{S}:=\{a_{\sigma(i), \sigma(j)}\}$ for each $i,j\in\mathcal{S}$ where $\sigma:\mathcal{S}\to \{1, 2, \dots, |\mathcal{S}|\}$ is the unique increasing bijection. A strongly connected component $\mathcal H$ is a subgraph of $\mathcal{G}$ that is strongly connected and such that any subgraph of $\mathcal{G}$ strictly containing $\mathcal{H}$ is not strongly connected.

Our objective is to select a control node set $\mathcal{C}$ such that $-D_1 + B_{11}$ is a Hurwitz Metzler matrix, thus fulfilling the condition of Theorem~\ref{thm:partial_infection}. The main tool we exploit here is \cite[Theorem 4.14]{duan2021graph}, which provides graph-theoretic conditions (involving cycle gains, described below) for an irreducible Metzler matrix to be Hurwitz. To begin, let $M$ be a Metzler matrix of dimension $n$ with negative diagonal elements, and  $\mathcal{G}(M) = (\mathcal{V}_M, \mathcal{E}_M, M)$. Let $\Phi$ be the set of simple cycles of $\mathcal{G}(M)$. Let $c\in \Phi$ be a simple cycle of length $h\geq 2$ that starts and ends at $i_1$, containing the edges $(i_1, i_2),(i_2, i_3),\hdots,(i_{h-1},i_h),(i_h, i_1)$. The sum-cycle gain of $c$ is defined as 
\begin{equation}\label{eq:sumcyclegain_definiton}
    \gamma_c := q(M, c),
\end{equation}
where the map $q:\mathbb{R}^{n\times n} \times \Phi\to \mathbb{R}_{\geq 0}$ is given by
\begin{equation}\label{eq:sumcyclegain_map_q}
    (M,c) \mapsto \Big( \frac{m_{i_2 i_1}}{-m_{i_2 i_2}}\Big)\Big( \frac{m_{i_3 i_2}}{-m_{i_3 i_3}}\Big) \dots \Big( \frac{m_{i_1 i_h}}{-m_{i_1 i_1}}\Big).
\end{equation}
We define the sum of the sum-cycle gains as
\begin{equation}\label{eq:sum_sumcycle_def}
    S := \sum_{c\in\Phi}\gamma_c.
\end{equation}
A key result is that a sufficient condition for an irreducible $M$ to be Hurwitz is that $S<1$, see \cite[Theorem 4.14]{duan2021graph}. 



To apply this result in our context,  we introduce a function \[[\eta, \gamma_\eta, S] = \textsc{cycleGains}(\mathcal{G})\]
which takes as an input an arbitrary graph $\mathcal{G} = (\mathcal{V, E}, A)$. The output is i) the cycle $\eta \in \Phi$ with the greatest sum-cycle gain, ii) the sum-cycle gain of cycle $\eta$, denoted by $\gamma_\eta$, and iii) the sum of the sum-cycle gains, $S$. Note that $\eta$ is not necessarily unique; if two or more cycles have the common greatest sum-cycle gain, we select $\eta$ at random. If $\mathcal{G}$ has no cycles, then we define $S = 0$, $\eta = \emptyset$, and $\gamma_\eta = 0$.

The proposed algorithm consists of two stages; in both, we take nodes from $\mathcal{V}$ and iteratively assign them to the controlled node set $\mathcal{C}$, so that $\mathcal{C}$ grows over the course of the algorithm. During this process, we define $\tilde{\mathcal U} = \mathcal{V} \setminus \mathcal{C}$ as the set of ``potentially uncontrolled nodes'', i.e., nodes that may   still be assigned to $\mathcal{C}$; at the end of the two stages, we will have finished assigning nodes to $\mathcal{C}$. Hence, we have also established the set of uncontrolled nodes $\mathcal{U} = \mathcal{V}\setminus \mathcal{C}$, which is characterised by the requirement that the sum of the sum-cycle gains of $\mathcal G[\mathcal U]$ is less than 1. 

In Stage 1, we assign to $\mathcal{C}$ all the nodes $i\in\mathcal{V}$ for which $d_i\leq b_{ii}$.

\begin{algorithm}
\caption{Stage 1}
\begin{algorithmic}[1]
\Procedure {buildC}{$D$, $B$, $\mathcal{V}$}
\State $\mathcal{C}\leftarrow\emptyset$
\ForAll{$i\in\mathcal{V}$}
    \If{$d_i \leq b_{ii}$}
        \State Assign node $i$ to $\mathcal{C}$
    \EndIf
\EndFor
\State $\tilde{\mathcal{U}}\leftarrow \mathcal{V\backslash C}$
\State $\mathcal{G}_{\tilde{\mathcal U}} \leftarrow \mathcal{G}[\tilde{\mathcal{U}}]$
\EndProcedure
\end{algorithmic}
\end{algorithm}

At the end of Stage~1, we have $\tilde{\mathcal{U}}$ and the induced subgraph $\mathcal{G}[\tilde{\mathcal{U}}]$. Stage 2 examines $\mathcal{G}[\tilde{\mathcal{U}}]$, progressively removing nodes from $\tilde{\mathcal{U}}$ and assigning them to $\mathcal{C}$. 
First, define
\[[\mathcal{P}_\mathcal{T}] = \textsc{SCC}(\mathcal{G}_{\tilde{\mathcal{U}}})\]
as the function which takes as input the graph $\mathcal{G}_{\tilde{\mathcal{U}}}$ and produces as an output a set $\mathcal{P}_{\tilde{\mathcal{U}}} :=\{\mathcal{P}_{\tilde{\mathcal{U}}}^1, \mathcal{P}_{\tilde{\mathcal{U}}}^2, \dots, \mathcal{P}_{\tilde{\mathcal{U}}}^r\}$ of strongly connected components. In particular, there are $r = \vert \mathcal{P}_{\tilde{\mathcal{U}}}\vert$ strongly connected components, and the $i$th strongly connected component is $\mathcal{P}_{\tilde{\mathcal{U}}}^i = (\mathcal{V}_{\tilde{\mathcal{U}}}^i, \mathcal{E}_{\tilde{\mathcal{U}}}^i, A_{\tilde{\mathcal{U}}}^i)$.


We now briefly describe the execution of Stage~2 for a generic component $\mathcal{P}_{\tilde{\mathcal{U}}}^i$; each component is considered in turn. We focus on the while-loop. If the sum of the sum-cycle gains $S_i<1$, we do nothing.  If the sum of the sum-cycle gains $S_i \geq 1$, then a node $j$ is randomly chosen from the cycle with the largest sum-cycle gain, $\eta_i$, removed from $\mathcal{P}_{\tilde{\mathcal{U}}}^i$ and assigned to $\mathcal{C}$ (lines 6--8). As a consequence, at least one cycle is broken (and possibly more if node $j$ belongs to multiple cycles). It is possible (but not necessarily the case) that $\mathcal{P}_{\tilde{\mathcal{U}}}^i$ loses its strong connectivity property when $\eta_i$ is broken, but this is not an issue as the \textsc{cycleGains} function does not require the input graph to be strongly connected. 
The updated component $\mathcal{P}_{\tilde{\mathcal{U}}}^i$ takes into account the removal of node~$j$ (line~$7$). 
Compared to that before the removal of node~$j$ (which results in cycle $\eta^i$ being broken), the number cycles in the updated component $\mathcal{P}_{\tilde{\mathcal{U}}}^i$ reduces by at least $1$, and furthermore, $S_i$ decreases by at least $\gamma_\eta^i$. The \textsc{cycleGains} function provides updated values of $\eta^i, \gamma_\eta^i, S^i$ for the updated component $\mathcal{P}_{\tilde{\mathcal{U}}}^i$ without node~$j$ (line~$9$). The while loop terminates the moment that removal of node~$j$ yields $S^i < 1$. 

Once every component in $\mathcal{P}_{\tilde{\mathcal{U}}}$ has been operated upon in Stage~2 (and thus we are at the end of Stage~2, line~12), we have the final controlled and uncontrolled node sets $\mathcal{C}$ and $\mathcal{U} = \mathcal{V}\setminus \mathcal{C}$, respectively. We conclude Section~\ref{sssec:algorithm_control} with the following result, which states that $\mathcal{U}$ is nonempty and the control problem is solved.

\begin{algorithm}
\caption{Stage 2}
\begin{algorithmic}[1]
\Procedure{addtoC}{$D$,$B$,$\mathcal{T}$}
\State $\mathcal{P}_{\tilde{\mathcal{U}}} = \{\mathcal{P}_{\tilde{\mathcal{U}}}^1, \mathcal{P}_{\tilde{\mathcal{U}}}^2, \dots, \mathcal{P}_{\tilde{\mathcal{U}}}^r\} \leftarrow$ \Call{SCC}{$\mathcal{G}_{\tilde{\mathcal{U}}}$} 
\ForAll{$i= 1, 2, \dots, r$}
    \State $\eta^i, \gamma_\eta^i, S^i \leftarrow$\Call{cycleGains}{$\mathcal{P}_{\tilde{\mathcal{U}}}^i$}
    \While{$S^i\geq 1$}
        \State Choose any node $j$ from $\eta^i$
        \State $\mathcal{P}_{\tilde{\mathcal{U}}}^i \leftarrow \mathcal{P}_{\tilde{\mathcal{U}}}^i[\mathcal{V}_{\tilde{\mathcal{U}}}^i\backslash\{j\}]$ (this removes $j$ from $\mathcal{P}_{\tilde{\mathcal{U}}}^i$)
        \State Assign $j$ to $\mathcal{C}$
        \State $\eta^i, \gamma_\eta^i, S^i\leftarrow$\Call{cycleGains}{$\mathcal{P}_{\tilde{\mathcal{U}}}^i$}
    \EndWhile
\EndFor
\EndProcedure
\end{algorithmic}
\end{algorithm}


\begin{proposition}\label{prop:algorithm}
    Consider the system in \eqref{eq:simp_contsystotal} under Assumptions~\ref{assm:strongly_connected}, \ref{assm:endemic_R0} and \ref{ass:phi_properties_partialcont}. Assume there exists $i\in\mathcal{V}$ such that $d_i>b_{ii}$. Then the proposed algorithm terminates with a controlled node set $\mathcal{C}\neq \mathcal{V}$ which solves Problem~\ref{prob:partialinfcontrol}.
\end{proposition}
\begin{proof}
    Obviously, at the termination of Stage 1, the set of potentially uncontrolled nodes $\tilde{\mathcal{U}}$  
    is nonempty due to the presence of the node $i\in\mathcal{V}$ with $d_i>b_{ii}$. 
    
    At the beginning of Stage~2, $\mathcal{G}_{\tilde{\mathcal{U}}}$ has $r$ strongly connected components, with $r \geq 1$. We now show that, for each strongly connected component $\mathcal{P}_{\tilde{\mathcal{U}}}^i$ (as identified in Line~2), the while loop of Stage~2 terminates with $\mathcal{V}_{\tilde{\mathcal{U}}}^i\neq \emptyset$. First, note that the while loop is skipped i) if there are no simple cycles in $\mathcal{P}_{\tilde{\mathcal{U}}}^i$ (and thus $S^i = 0$) or ii) if $S^i<1$. In this case, Stage~2 terminates, and $\mathcal{V}_{\tilde{\mathcal{U}}}^i\neq \emptyset$ since no nodes have been removed from it. 
    
    Suppose then, that there is at least one cycle in $\mathcal{P}_{\tilde{\mathcal{U}}}^i$ at the start of the while loop, and that $S^i\geq 1$. As noted above, removal of node~$j$ breaks cycle $\eta^i$, and possibly other cycles that $j$ belongs to.
    Let $\{1, 2, \hdots, q\}$, with $q\geq 1$, be the set of cycles of $\mathcal{P}_{\tilde{\mathcal{U}}}^i$ before the removal of the final node $j$ that leads to the termination of the while loop. From the definition of a simple cycle, each cycle has at least two nodes. Thus, there must remain at least one other node in $\mathcal{P}_{\tilde{\mathcal{U}}}^i$ after the removal of node~$j$ terminates the while loop.
    
    Without loss of generality, at the end of Stage~2, order the nodes as $\mathcal{U} = \{1, \hdots, k\}$ and $\mathcal{C} = \{k+1, \hdots, n\}$, with $B$ and $D$ partitioned as in \eqref{eq:DB_partition}. We have just established that there are $k \geq 1$ nodes in $\mathcal{U}$. We complete the proof by showing that $-D_1 + B_{11}$ is Hurwitz at the end of Stage~2. For convenience, we define $A = -D_1+B_{11}$ and reorder the nodes in $\mathcal{U}$ so that we can write $A$ in the block lower-triangular form:
    \begin{equation}
        A = \begin{bmatrix}
            A^{11} & \vect 0 & \cdots & \vect 0\\
            A^{12} & A^{22} & \ddots & \vect 0  \\
            \vdots & \ddots & \ddots & \vdots\\ 
            A^{1r} & \cdots & A^{(r-1)r} & A^{rr}
        \end{bmatrix}
    \end{equation}
    Note that $A^{ii}$ for $i = 1,2 \hdots, r$ correspond to the subgraph induced by $\mathcal{P}_\mathcal{U}^i$ at the end of Stage~2 of the algorithm; we demonstrate $A$ is Hurwitz by showing every $A^{ii}$ is Hurwitz. 
    
    Towards this end, consider each $\mathcal{P}_\mathcal{U}^i$ after Stage~2, with $A^{ii}$ the associated Metzler matrix having negative diagonal entries. As noted above the proposition, $S^i < 1$, and $\mathcal{P}_\mathcal{U}^i$ may or may not be strongly connected (and hence $A^{ii}$ may or may not be irreducible). 
    Without loss of generality, reorder the nodes in $\mathcal{P}_\mathcal{U}^i$ so that $A^{ii}$ is in a block lower-triangular form; if $\mathcal{P}_\mathcal{U}^i$ is strongly connected then $A^{ii}$ is irreducible, and otherwise each diagonal block corresponds to a strongly connected component of $\mathcal{P}_\mathcal{U}^i$. The cycles of $\mathcal{P}_\mathcal{U}^i$ (whose sum-cycle gains add up to $S^i < 1$) are the cycles of its strongly connected components. It follows that the sum of the sum-cycle gains of each strongly connected component is strictly less than $1$, and hence the associated diagonal block of $A^{ii}$ is Hurwitz~\cite[Theorem 4.14]{duan2021graph}. Since every diagonal block of $A^{ii}$ is Hurwitz, $A^{ii}$ itself must be Hurwitz. As this conclusion holds for every $A^{ii}$, it follows that $A = -D_1 +B_{11}$ is Hurwitz. Thus, the given $\mathcal{U}$ and $\mathcal{C}$ satisfy the hypotheses of Theorem~\ref{thm:partial_infection}, and Problem~\ref{prob:partialinfcontrol} is solved.
\end{proof}

It should be noted that, for a given network with pair $(D,B)$, there may be multiple node sets which satisfy the conditions for controlling the SIS network. Our algorithm will ensure that a suitable node set $\mathcal{C}$ is always found, and assuming $d_i > b_{ii}$ for some $i\in\mathcal{V}$, then $\mathcal{U}$ will always be nonempty. Our algorithm may produce different $\mathcal{C}$ sets each time it is executed, due to the random selection of $\eta$ in the $\textsc{cycleGains}$ function, and the random selection of node~$j$ in Line~6 of Stage~2. Finding a \textit{minimal set} of nodes to control, viz. minimising $\vert \mathcal{C}\vert$, is a significantly more challenging problem. Another challenge is to reduce the computational complexity of the proposed algorithm; our approach requires iteratively finding cycles in subgraphs, which can be expensive for dense networks. Solving these two challenges is beyond the scope of this paper, and we leave it for future research. 

\subsection{Partial Recovery Rate Control Problem}\label{ssec:partial_recovery}
Similarly to Section~\ref{sec:full_control}, we can consider a partial recovery rate control problem as a complementary approach to the partial infection rate control problem. Indeed, we can adopt the same Assumption~\ref{ass:phi_properties_partialcont} and consider a problem statement which is identical to Problem~\ref{prob:partialinfcontrol} except we study system \eqref{eq:simp_cont_recovery} and require $\lim_{t\to\infty} g_i(t) < \infty$ for all $i\in\mathcal{V}$. Due to similarity with the proofs in Section~\ref{ssec:partial_infection}, we state the main result of Section~\ref{ssec:partial_recovery} here without proof. 

\begin{theorem}\label{thm:partial_recovery}
    Consider the system in \eqref{eq:simp_cont_recovery} under Assumptions~\ref{assm:strongly_connected}, \ref{assm:endemic_R0} and \ref{ass:phi_properties_partialcont}. Without loss of generality, let the nodes be ordered as $\mathcal{T} = \{1, \hdots, k\}$ and $\mathcal{U} = \{k+1, \hdots, n\}$, with $B$ and $D$ partitioned as in \eqref{eq:DB_partition}. Then the following statements are equivalent.
    \begin{enumerate}
        \item For all $\xi(0)\in \Xi_n \times \vect 1_n$, there holds $\lim_{t\to\infty} x(t) = \vect 0_n$ and $\lim_{t\to\infty} g_i(t) = \bar g_i$, where $\bar g < \infty$ for all $i\in\mathcal{V}$.
        \item The matrix $- D_1 + B_{11}$ is Hurwitz.
    \end{enumerate}
\end{theorem}

The necessary and sufficient condition (first statement) and sufficient condition (second statement) of Proposition~\ref{prop:necessary_partial} are identical. The algorithm in Section~\ref{sssec:algorithm_control} can also be used to find a suitable controlled node set $\mathcal{U}$, and we can similarly guarantee that the algorithm will terminate with nonempty $\mathcal{T}$ provided there is some $i\in\mathcal{V}$ such that $d_i > b_{ii}$ (see Proposition~\ref{prop:algorithm}). The auxiliary results in Section~\ref{sec:full_control}, namely Propositions~\ref{lem:gaininequality}, \ref{prop:positive_finite_t}, and \ref{prop:final_reproduction} continue to hold for the partially controlled network system, both for infection rate control and recovery rate control.

\begin{remark}
The adaptive algorithms proposed in \eqref{eq:adaptive_law} and \eqref{eq:adaptive_law_recov}  are decentralised so that each node can execute the algorithm independently of other nodes. Moreover, algorithm execution does not require knowledge of the network (i.e., knowledge of the infection and recovery rates, and the network structure). Thus, in the full network control scenario, our method is fully decentralised and requires no knowledge of the network. The drawback is that every node must be controlled, which may be expensive in large-scale networks. This led us to consider the partial network control scenario, where our method continues to be decentralised and requires no knowledge of the network during execution. However, the trade-off for controlling just a subset of nodes is that we require i) a centralised iterative algorithm to select the controlled nodes, and ii) knowledge of the network to verify the condition in Theorems \ref{thm:partial_infection} and \ref{thm:partial_recovery} and to run the iterative algorithm. \hfill $\triangle$
\end{remark}


\subsection{Simulations for Partial Network Control}\label{ssec:sim_partialnetwork}
We conclude Section~\ref{sec:partial_control} by considering a toy example with $n = 6$ nodes that allows us to more easily see how the network structure can influence the partial network control problem. 


\begin{figure}
\centering
{\def\svgwidth{0.6\linewidth}
	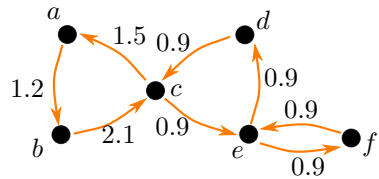}
    \caption{Network structure of example in Section~\ref{ssec:sim_partialnetwork}.  }
	\label{fig:example_network}
\end{figure}

For clarity, we label the six nodes alphabetically, $a$, $b$, $c$, $d$, $e$, $f$. We set $d_{i} = 2$ for all $i \in \{a,b,c,\hdots, f\}$, and the network topology associated with the infection transmission is given in Fig.~\ref{fig:example_network}. Self-loops are not drawn for clarity, but each node has a self-loop with weight $1$, except node~$a$ which has a self-loop with weight $4$. 

Using Stage~1 of our proposed algorithm, we establish that we must control node $a$, i.e., $a \in \mathcal{C}$; this is because $d_{a} < b_{aa}$. (Note this is also in accordance with Proposition~\ref{prop:necessary_partial}). At the start of Stage~2 of our algorithm, we have $\tilde{\mathcal{U}} = \{b, c, d, e, f\}$, and the associated $-D_1+B_{11}$ matrix has a spectral abscissa of $0.1922$, i.e., the matrix is not Hurwitz; this indicates that at least one node of $\tilde{\mathcal U}$ must be transferred to the set of controlled nodes $\mathcal C$. Such a node is determined as follows. There are two strongly connected components: $\mathcal{P}_{\tilde{\mathcal{U}}}^1$ with node set $\mathcal{V}_{\tilde{\mathcal{U}}}^1 = \{c,d,e,f\}$ and $\mathcal{P}_{\tilde{\mathcal{U}}}^2$ with node set $\mathcal{V}_{\tilde{\mathcal{U}}}^2 = \{b\}$. Since $\mathcal{P}_{\tilde{\mathcal{U}}}^2$ has no simple cycles, we do not need to move any nodes from $\mathcal{V}_{\tilde{\mathcal{U}}}^2$ to $\mathcal{C}$. For $\mathcal{P}_{\tilde{\mathcal{U}}}^1$, there are two simple cycles $\{(c,e),(e,d),(d,c)\}$ and $\{(e,f),(f,e)\}$, with sum-cycle gains $0.729$ and $0.81$, respectively. Thus, $S^1 \geq 1$. As it turns out, removing any node from $\mathcal{V}_{\tilde{\mathcal{U}}}^1$ breaks at least one of the two cycles and the resulting $S^i < 1$, which terminates Stage~2 of the algorithm. One can easily check that at the end of Stage~2, the $-D_1+B_{11}$ matrix associated with any of the possible resulting $\mathcal{U}$ is Hurwitz.

In our simulations, we sample $x_i(0)$ from a uniform distribution $(0,1)$, and we consider the partial infection rate control problem. For any node $i\in\mathcal{C}$, we set $\phi_i = x_i$, i.e., $\alpha_i = p = 1$. In Fig.~\ref{fig:n6_partialinfection_U14}, we set $\mathcal{C} = \{a,d\}$. In Fig.~\ref{fig:n6_partialinfection_U16}, we set $\mathcal{C} = \{a,f\}$, and in Fig.~\ref{fig:n6_partialinfection_U1}, we set $\mathcal{C} = \{a\}$. We can see that for both $\mathcal{C} = \{a,d\}$ and $\mathcal{C} = \{a,f\}$, control of just two nodes is sufficient to eliminate the disease from the entire network while ensuring the adaptive gains of the controlled nodes converge to strictly positive values. However, notice that the rate of convergence differs significantly depending on whether node~$d$ or node~$f$ is controlled, with up to an order of magnitude difference. This suggests that the network structure (and the associated matrices $D$ and $B$) play a highly nontrivial role in shaping the controlled dynamics. If we only control node~$a$, we see in Fig.~\ref{fig:n6_partialinfection_U1} that the disease is eliminated from nodes~$a$ and $b$, but remains endemic in nodes $c,d,e,f$, and $\lim_{t\to\infty} g_a(t) = 0$.

\begin{figure*} 
\centering
      \subfloat[Network dynamics, $\mathcal{U} = \{a,d\}$]{\includegraphics[width= 0.32\linewidth]{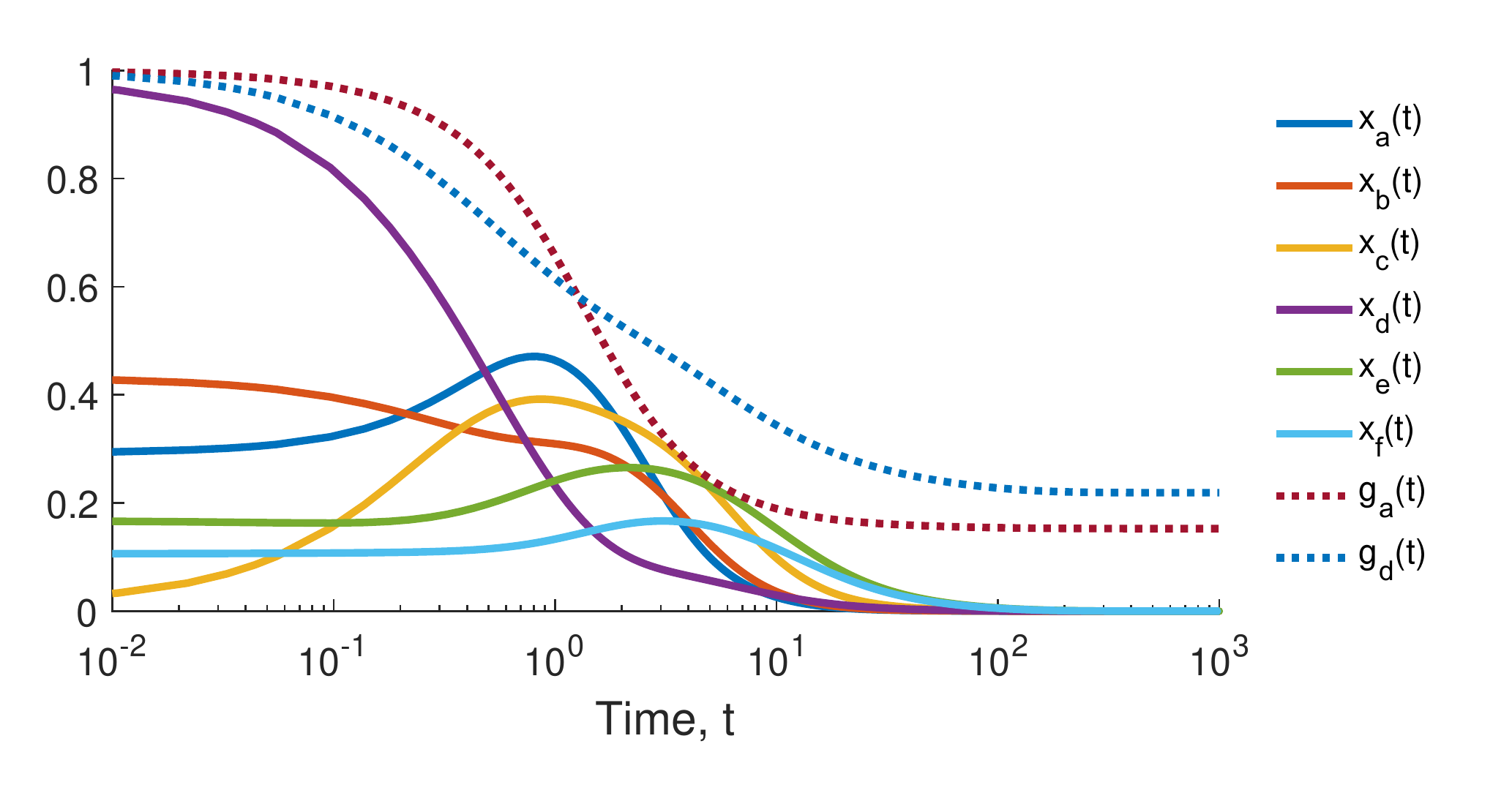}\label{fig:n6_partialinfection_U14}}
      \subfloat[Network dynamics, $\mathcal{U} = \{a,f\}$]{\includegraphics[width= 0.32\linewidth]{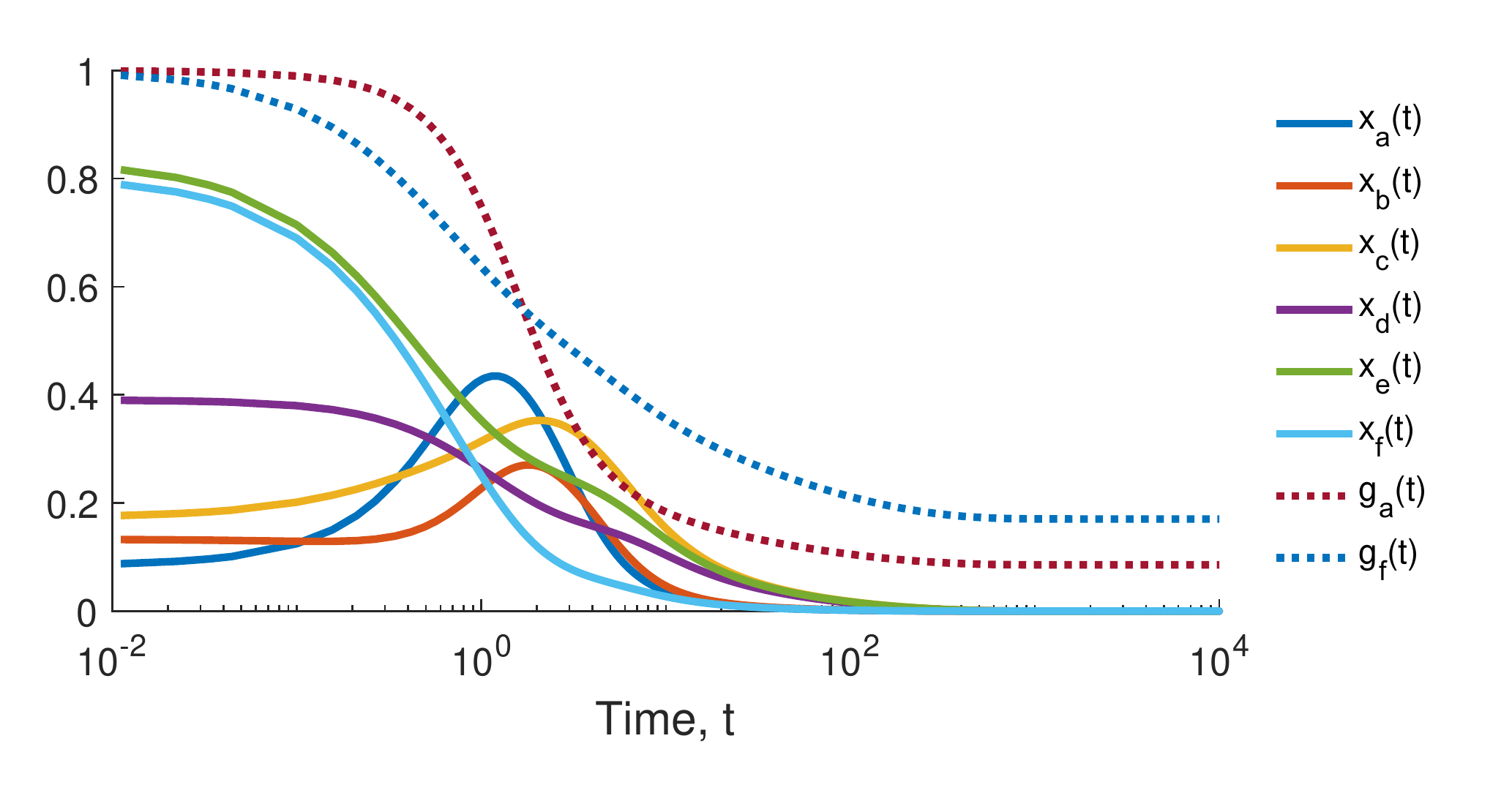}\label{fig:n6_partialinfection_U16}}
      \subfloat[Network dynamics, $\mathcal{U} = \{a\}$]{\includegraphics[width= 0.32\linewidth]{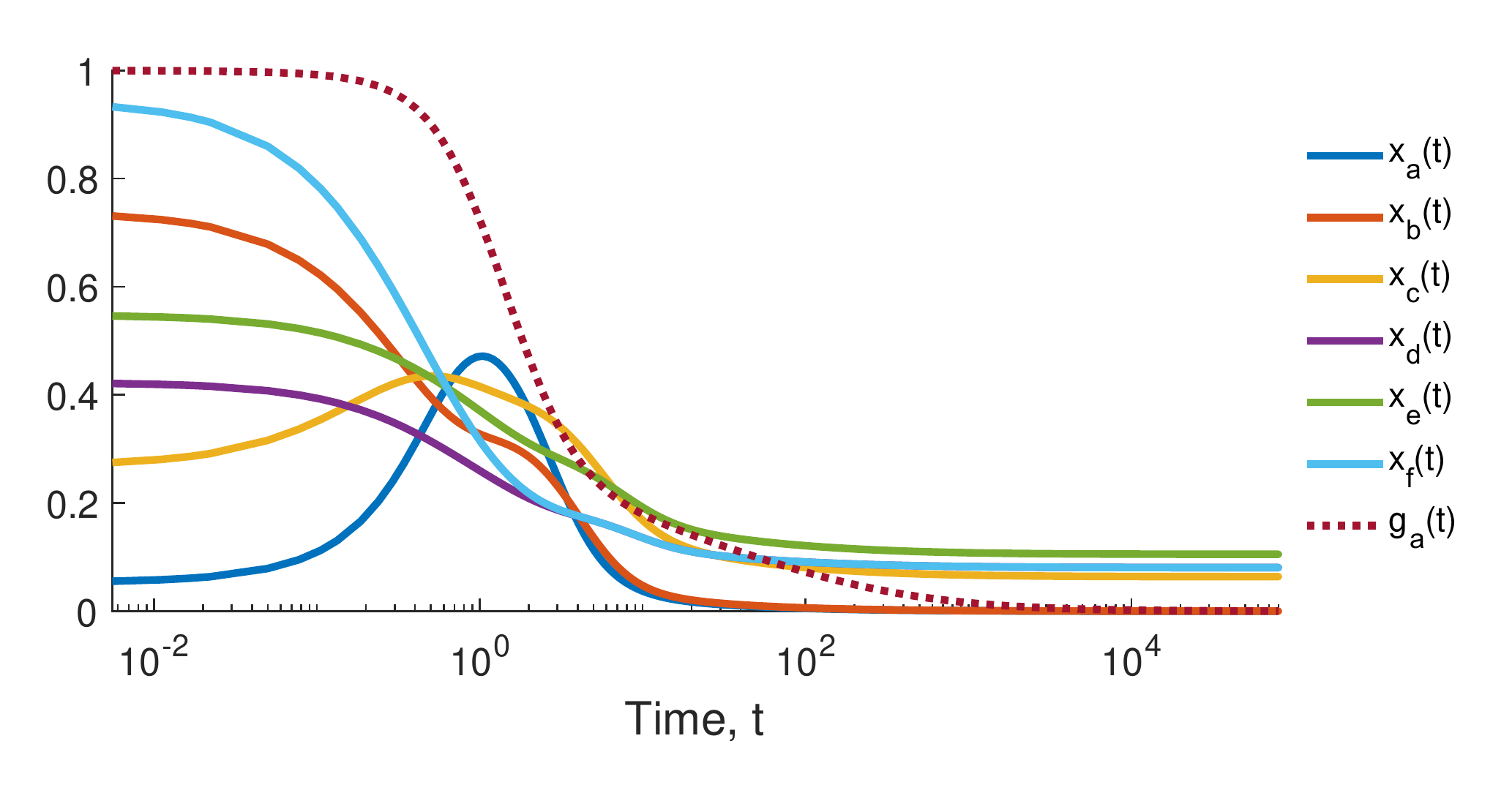}\label{fig:n6_partialinfection_U1}}
    \caption{Decentralised adaptive-gain partial infection rate control for Fig.~\ref{fig:example_network}. For clarity, we only show the time evolution of infection fraction $x_i(t)$ (solid coloured lines) and gains $g_i(t)$ from the controlled nodes (dotted lines). Note the logarithmic scale of $t$ on the horizontal axis.}    \label{fig:partial}
\end{figure*}

\section{Conclusion}\label{sec:con}
This paper considered a suite of feedback control problems for eliminating the spread of an infectious disease, described by the SIS network epidemic model. Decentralised adaptive-gain algorithms were proposed to control the infection rates and recovery rates at each node, and we considered both controlling i) all nodes in the network, and ii) a partial subset of the nodes. The proposed algorithms are able to drive the network to the healthy equilibrium, while ensuring the gains remain positive and finite. 

A number of directions for future work should be considered. First, one should investigate piecewise constant updating of the gain, instead of continuous updating, to better reflect real-world interventions which are rolled out in phases. This could either occur via periodic updating (which is very likely quite straightforward), or by an event-triggered approach. Second, one can consider a combination of controlling the recovery rate for some nodes, and the infection rate for other nodes, or, in another direction one might seek to control selected edges only (edge-based network control). Third, we would like to devise more sophisticated adaptive algorithms, which allow restoration of the gain towards $1$ (i.e. the initial gain) when the disease is close to being eliminated. Finally, our results suggest $p = 1$ yields a faster convergence rate than $p > 1$, but further rigorous examination would help clarify this.

\bibliographystyle{IEEEtran}        
\bibliography{Walsh}

\end{document}